\documentclass[a4paper,UKenglish]{lipics-v2016}

\usepackage{mathtools}
\usepackage{amssymb}
\usepackage{cite}
\usepackage{tikz}
\usepackage{graphicx}
\usepackage{paralist}
\usepackage[basic]{complexity}
\usepackage{xspace}
\usepackage{wrapfig}

\makeatletter
\let\epsilon\varepsilon
\let\phi\varphi
\makeatother

\theoremstyle{definition}

\newtheorem{claim}[theorem]{Claim}

\usetikzlibrary{positioning,fit,arrows,automata,calc,shapes} 
\tikzset{->,>=stealth',shorten >=1pt,shorten <=1pt,auto,node distance=1cm,
every loop/.style={looseness=6},
initial text={},
el/.style={font=\scriptsize},
every fit/.style={draw,densely dotted,rectangle},
inner sep=2mm,
loopright/.style={loop,looseness=6,out=-45, in=45},
loopleft/.style={loop,looseness=6,out=135, in=225},
loopabove/.style={loop,looseness=6,out=45, in=135},
loopbelow/.style={loop,looseness=6,out=-135, in=-45},
}

\providecommand\st{}
\renewcommand{\st}{\:\mid\:}
\newcommand{\defeq}{\mathrel{\overset{\makebox[0pt]{\mbox{\normalfont\tiny\sffamily def}}}{=}}}
\newcommand{\defiff}{\mathrel{\overset{\makebox[0pt]{\mbox{\normalfont\tiny\sffamily
def}}}{\iff}}}
\newcommand{\pow}{\mathcal{P}}
\newcommand{\ie}{\textit{i.e.}\xspace}
\newcommand{\eg}{\textit{e.g.}\xspace}

\newcommand{\supfun}{\mathsf{Sup}}
\newcommand{\inffun}{\mathsf{Inf}}
\newcommand{\lsupfun}{\mathsf{LimSup}}
\newcommand{\linffun}{\mathsf{LimInf}}
\newcommand{\mpfun}{\mathsf{MP}}

\def\newdef{\emph}
\newcommand{\G}{\texttt{G}}
\newcommand{\F}{\texttt{F}}
\newcommand{\X}{\texttt{X}}
\newcommand{\U}{\mathrel{\texttt{U}}}
\newcommand{\calG}{\mathcal{G}}
\newcommand{\calA}{\mathcal{A}}

\newcommand{\calD}{\mathcal{D}}
\newcommand{\lcp}{\ensuremath{\mathrm{lcp}}}

\newcommand{\last}{\ensuremath{\mathrm{last}}}
\newcommand{\first}{\ensuremath{\mathrm{first}}}
\newcommand{\prefix}{\ensuremath{\subseteq_{\mathsf{pref}}}}
\newcommand{\vinit}{{\ensuremath{v_{\textrm{\sf init}}}}}

\newcommand{\payoff}[1]{\textrm{\sf payoff}_{#1}}
\newcommand{\pfunction}[1]{{\payoff{#1}}}
\newcommand{\vwdom}[1]{\ensuremath{\succcurlyeq_{#1}}}
\newcommand{\wdom}[1]{\ensuremath{\succ_{#1}}}
\newcommand{\switch}[2]{\ensuremath{\left[#1\leftarrow#2\right]}}

\def\Agt{\ensuremath{P}}
\newcommand{\outcome}{\ensuremath{\mathbf{Out}}}
\newcommand{\history}{\ensuremath{\mathbf{Hist}}}
\newcommand{\stratset}{\ensuremath{\Sigma}}
\newcommand{\admstratset}{\ensuremath{\mathfrak{A}}}

\newcommand\player[1]{player~\ensuremath{#1}\xspace}
\newcommand\Player[1]{Player~\ensuremath{#1}\xspace}
\newcommand{\aVal}{\mathbf{aVal}}
\newcommand{\acVal}{\mathbf{acVal}}
\newcommand{\cVal}{\mathbf{cVal}}

\def\Gaih{\ensuremath{\calG^{\geq \aVal_i(h)}}}

\def\wco{\textrm{\sf wco}}
\def\sco{\textrm{\sf sco}}
\newcommand\propadm[2]{\star(#1,#2)}
\def\ltlpayoff{\ensuremath{\textrm{\sf \LTL{}}_\textrm{\sf payoff}}}
\def\LTL{\textrm{\sf LTL}}
\newcommand{\labeling}{\lambda}
\newcommand{\ap}{\mathsf{AP}}
\newcommand{\gAlt}{\mathsf{gAlt}^i}
\newcommand{\aValProp}{\mathsf{aVal}^i}
\newcommand{\acValProp}{\mathsf{acVal}^i}
\newcommand{\ownProp}{\mathsf{V_i}}

\newcommand*{\Cdot}{\raisebox{.8ex}{$\smallfrown$}}

\title{Admissibility in Quantitative Graph Games\footnote{This work was partially
supported by the ERC Starting Grant inVEST (279499) and EPSRC grant EP/M023656/1.}}

\author[1]{Romain Brenguier}
\author[2]{Guillermo A. P\'{e}rez\thanks{Author supported by F.R.S.-FNRS
fellowship.}}
\author[2]{Jean-Fran\c{c}ois Raskin}
\author[3]{Ocan Sankur}
\affil[1]{University of Oxford, Oxford, UK}
\affil[2]{Universit\'{e} Libre de Bruxelles (ULB), Brussels, Belgium}
\affil[3]{CNRS, Irisa, Rennes, France}
\authorrunning{R. Brenguier, G. A. P\'{e}rez, J.-F. Raskin, and O. Sankur}

\Copyright{Romain Brenguier, Guillermo A. P\'{e}rez, Jean-Fran\c{c}ois Raskin,
Ocan Sankur}

\begin{document}

\maketitle

\begin{abstract}
	{\em Admissibility}
	has been studied  for games
	of infinite duration with Boolean objectives. We extend here
	this study to games of infinite duration with {\em quantitative} objectives.
	First, we show that, under the assumption that optimal worst-case and
	cooperative strategies exist, admissible strategies are guaranteed to
	exist. Second, we give a characterization of admissible strategies using
	the notion of adversarial and cooperative values of a history, and we
	characterize the set of outcomes that are compatible with admissible
	strategies. Finally, we show how these characterizations can be used to
	design algorithms to decide relevant verification and synthesis problems. 
\end{abstract}

\section{Introduction}

 Two-player {\em zero-sum} graph games are the most studied mathematical model
 to formalize the reactive synthesis problem~\cite{PnRo89,Thomas95}.
 Unfortunately, this mathematical model is often an abstraction that is too
 coarse.  Realistic systems are usually made up of several components, all of
 them with their {\em own} objectives. These objectives are not necessarily
 antagonistic.  Hence, the setting of {\em non-zero sum} graph games is now
 investigated in order to {\em unleash} the full potential of automatic
 synthesis algorithms for reactive systems, see
 \eg~\cite{KHJ06,berwanger07,BRS14,CDFR14,KPV14,FismanKL10}.

For a player with objective $\phi$, a
strategy $\sigma$ is said to be {\em dominated} by a strategy $\sigma'$ if $\sigma'$
does as well as $\sigma$ with respect to $\phi$ against all the strategies of
the other players and strictly better for some of them. A strategy $\sigma$ is {\em admissible}
for a player if it is {\em not} dominated by any other of his strategies.
Clearly, playing a strategy which is not admissible is sub-optimal and a
{\em rational} player should only play admissible strategies. 
The elimination of dominated strategies can be
{\em iterated} if one assumes that each player knows
the other players know that only admissible strategies are played, and so on.

While admissibility is a classical notion for finite games in normal form, see
\eg~\cite{Gretlein1983} and pointers therein, its generalization to infinite
duration games is challenging and was only considered more recently. In 2007,
Berwanger was the first to show~\cite{berwanger07} that admissibility, \ie the
avoidance of dominated strategies, is well-behaved in infinite duration
$n$-player non-zero sum turn-based games with perfect information and Boolean
outcomes (two possible payoffs: win or lose). This framework encompasses games
with omega-regular objectives. The main contributions of Berwanger were to show
that $(i)$ in all $n$-player game structures, for all objectives, players have
{\em admissible strategies},
(Berwanger even shows the existence of strategies
that survive the iterated elimination of strategies)
$(ii)$ every strategy
that is dominated by a strategy is dominated by an admissible strategy, and
$(iii)$ for finite game structures, 
the set of
admissible strategies forms a regular set.

While the iterated admissibility formalizes a strong notion of
rationality~\cite{adam2008admissibility}, it has been shown recently that the
non-iterated version is strong enough to synthesize relevant strategies for
non-zero sum games of infinite duration modelling reactive systems.
In~\cite{Faella09}, Faella considers games played on finite graphs and focuses
on the states from which one designated player cannot force a win. He compares
several criteria for establishing what is the preferable behavior of this player
from those states, eventually settling on the notion of admissible strategy.
In~\cite{brs15}, starting from the notion of admissible strategy, we have
defined a novel rule for the compositional synthesis of reactive systems,
applicable to systems made of $n$ components which have each their own
objective. We have shown that this synthesis rule leads to solutions which are
robust and resilient.

Here, we study the notion of admissible strategy in infinite horizon
$n$-player turn-based \emph{quantitative} games played on a finite game
structure. We give a comprehensive picture of the properties related to the
existence of such strategies and to their characterization.  Contrary to the Boolean case, the number of payoffs
in our setting is potentially infinite making the characterization challenging. As in~\cite{berwanger07}, we assume all players have
perfect information. 

\subparagraph{Main contributions.} First, contrary to the Boolean case, we
show that in the quantitative setting, there are dominated strategies that are
not dominated by any admissible strategy (Example~\ref{exa:infinite-chain}).
Second, we show that the existence of worst-case optimal and cooperatively
optimal strategies for all players is a sufficient condition for
the existence of admissible strategies (Thm.~\ref{thm:existence}). Additionally, we show that there are games without worst-case optimal or without cooperative optimal strategies that do not have admissible strategies (Lem.~\ref{lem:sometimes-no-adm}). Third,
we provide a characterization of admissible strategies in terms of antagonistic
and cooperative values---that are classical values defined for quantitative
games---(Thm.~\ref{thm:lhd-admissibles}) and a characterization of the
outcomes compatible with admissible strategies
(Thm.~\ref{thm:ltl-characterization}).  While the first characterization
allows one to precisely describe admissible strategies, the characterization of
the set of outcomes is given in linear temporal logic, and is a useful tool to
reason about the outcomes that can be generated by such strategies.
Finally, we show how to use the aforementioned characterizations to
obtain algorithms to solve relevant decision problems for games with classical
quantitative measures such as $\inffun$, $\supfun$, $\linffun$, $\lsupfun$ and
mean-payoff (Thms.~\ref{thm:strat-admissible}, \ref{thm:mc-admissible}, and
\ref{thm:aa-synthesis}).

\subparagraph{Example.}
Let us consider the game from Fig.~\ref{fig:example2} to illustrate several
notions and decision problems introduced and solved in this paper. The game is
played by two players: Player $1$, who owns the square vertices, and Player
$2$, owner of the round vertices. The measure that we consider here is the
mean-payoff. (But note that, the arguments we will develop in this example are
applicable to the limit inferior and limit superior measures as well.)

First, we note that the (best) worst-case value (or, the antagonistic value)
that Player $1$ can force is equal to $1$, while the antagonistic value for
Player $2$ is equal to $0$. The latter values are meaningful under the hypothesis
that the other player is playing fully antagonistically and not pursuing their
own objective. Now, if we account for the fact that Player~2 aims at maximizing his own payoff and so plays
only admissible strategies towards this goal, then we conclude
that he will never play the edge $(v_2,v_1)$. This is because, from vertex
$v_2$, Player $2$ has a strategy to enforce value $2$ and taking edge
$(v_2,v_1)$ is unreasonable because, in the worst case, from $v_1$ he will only
obtain $0$. As we show in Sec.~\ref{sec:applications}, this kind of reasoning
can be made formal and automated. We will show that, for games with classical quantitative measures,
 it can indeed be decided algorithmically if a finite memory strategy given, for instance, as a finite
state Moore machine, is admissible or not.

Second, a similar but more subtle reasoning to the one presented above allows us
to conclude that Player $1$ will eventually play the edge $(v_1,v_2)$. Indeed,
from vertex $v_1$, Player $1$ can force a payoff equal to $1$ by either taking
edge $(v_1,v_3)$ or $(v_1,v_2)$. Nevertheless, it is not reasonable for him to
play edge $(v_1,v_3)$ because, while this choice enforces a worst-case payoff
equal to $1$ (the antagonistic value), playing edge $(v_1,v_2)$ is better
because it ensures the same worst-case payoff and additionally leaves a
possibility for Player $2$ to help him by taking the cycle $v_2$--$v_4$, giving
him a payoff of $2$.
If we take into account that the adversary is playing
admissible strategies, then, in the words of~\cite{brs15}, we can solve the
assume-admissible synthesis problem. In this example, we conclude that Player
$1$ has a strategy to enforce a payoff of $2$ against all admissible strategies
of Player $2$. A strategy which eventually chooses edge $(v_1,v_2)$ ensures this
payoff. The formalization of this reasoning and elements necessary for its
automation are presented in Sec.~\ref{sec:applications}.
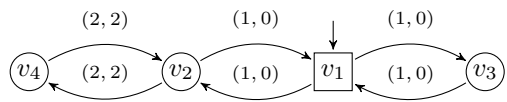
\begin{figure}
	\centering
	\begin{tikzpicture}
    \tikzstyle{every state}=[node distance=2cm,minimum size=5mm, inner sep=1pt];
    \node[state,initial above,rectangle] (s1) {$v_1$};
    \node[state,circle,right of=s1] (s3) {$v_3$};
    \node[state,circle,left of=s1] (s2) {$v_2$};
    \node[state,circle,left of=s2] (s4) {$v_4$};
    \path
    (s1) edge[bend left] node[el]{$(1,0)$} (s3)
    (s3) edge[bend left] node[el,swap]{$(1,0)$} (s1)
    (s1) edge[bend left] node[el,swap]{$(1,0)$} (s2)
    (s2) edge[bend left] node[el]{$(1,0)$} (s1)
    (s2) edge[bend left] node[el,swap]{$(2,2)$} (s4)
    (s4) edge[bend left] node[el]{$(2,2)$} (s2)
    ;
  \end{tikzpicture}
  \caption{Player~1 controls the square vertices, and Player 2 the round
	vertices. The payoff of Player~$i$ is the mean-payoff of the dimension~$i$
	of the weights seen along the run.}
  \label{fig:example2}
 \end{figure}

\subparagraph{Structure of the paper.}
Sec.~\ref{sec:preliminaries} contains definitions.
In Sec.~\ref{sec:existence},
we study conditions under which the existence of
admissible strategies is guaranteed. In
Sec.~\ref{sec:val-characterization}, we give a characterization
of admissible strategies, and in 
Sec.~\ref{sec:ltl-characterization}, a description of the set of
outcomes compatible with admissible strategies.
In Sec.~\ref{sec:applications}, we apply our results to solve relevant decision
problems on games with classical quantitative measures.

\section{Preliminaries}\label{sec:preliminaries}
We denote by
$\mathbb{R}$ the set of \newdef{real numbers}, $\mathbb{Q}$ the set
of \newdef{rational numbers}, $\mathbb{N}$ the set of \newdef{natural
numbers}, and $\mathbb{N}_{>0}$ the set of \newdef{positive
integers}.

	A \newdef{game}  is a tuple $\calG =
	\langle \Agt, $ $(V_i)_{i\in\Agt}, \vinit, E, (\pfunction{i})_{i\in\Agt}
	\rangle$ where:
	\begin{inparaenum}[$(i)$]
		\item $\Agt$ is the non-empty and finite set of players.
		\item $V \defeq \biguplus_{i\in \Agt} V_i$ where for every $i
			\in \Agt$, $V_i$ is the finite set of \player{i}'s
			vertices, and~$\vinit \in V$ is the \newdef{initial
			vertex}.
		\item $E \subseteq V \times V$ is the set of edges ({it
			is assumed, w.l.o.g., that each vertex in $V$ has at
			least one outgoing edge.})
		\item For every $i$ in $\Agt$, $\pfunction{i}$ is a
			\newdef{payoff function} from infinite paths in
			the digraph $\langle V,E \rangle$ to 
			$\mathbb{R}$ that, intuitively, \player{i} will attempt
			to maximize.
	\end{inparaenum}

An \newdef{outcome}~$\rho$ is an infinite path in the digraph $\langle V,E\rangle$,
\ie~an infinite sequence of vertices
$(\rho_j)_{j \in \mathbb{N}_{>0}}$ such that $(\rho_j,\rho_{j+1}) \in E$, for
all $j \in \mathbb{N}_{>0}$.
A finite prefix of an outcome is called a \emph{history}.  The
\newdef{length}~$|h|$ of a history $h = (\rho_j)_{1 \le j \le n}$ is $n$.  Given
an outcome~$\rho= (\rho_j)_{j \in \mathbb{N}_{>0}}$ and an integer $k$, we write
$\rho_{\le k}$ for the history $(\rho_j)_{1\le j\le k}$, that is, the prefix of
length $k$ of~$\rho$.  For a history $h$ and a history or outcome~$\rho$, we
write  $h \prefix \rho$ if~$h$ is a prefix of~$\rho$.  If $h \prefix \rho$, we
write $h^{-1} \cdot \rho$ for the unique history (resp. outcome) that satisfies
$\rho = h \cdot (h^{-1} \cdot \rho)$.  The \newdef{first} (resp.  \newdef{last})
vertex of a history $h$ is $\first(h) = h_1$ (resp.  $\last(h) \defeq
h_{|h|}$).  The \newdef{longest common prefix} of two outcomes or
histories~$\rho,\rho'$ is denoted $\lcp(\rho,\rho')$.  Given vertex $v$ from
$\calG$, let us denote the set of \newdef{successors of $v$} by $E_v \defeq \{
v' \in V \st (v,v') \in E\}$.  

	A \newdef{strategy} of \player{i} is a function $\sigma_i$ that maps any
	history $h$ such that $\last(h) \in V_i$ to a vertex from $E_{\last(h)}$.
	A \newdef{strategy profile} for the set of 
	players~$P' \subseteq \Agt$ is a tuple of strategies, one for each
	player of~$P'$.

Let $\stratset_{i}(\calG)$ be the set of all strategies of \player{i} in
$\calG$.  We write $\stratset(\calG) \defeq \prod_{i\in\Agt}\stratset_i(\calG)$
for the set of all strategy profiles for $\Agt$ in $\calG$, and
$\stratset_{-i}(\calG)$ for the set of strategy profiles for all players but $i$
in $\calG$. We omit $\calG$ when it is clear from the context.
Given $\sigma_i \in \stratset_i$ and $\sigma_{-i} = (\sigma_j)_{j\in
\Agt\setminus\{i\}} \in \stratset_{-i}$, we write
$(\sigma_i,\sigma_{-i})$ for $(\sigma_j)_{j\in \Agt}$.

A strategy profile $\sigma_\Agt \in \stratset$ defines a unique \newdef{outcome}
from any given history~$h$. Formally,
$\outcome_h(\calG,\sigma_\Agt)$ is the outcome $\rho
= (\rho_j)_{j \in \mathbb{N}_{>0}}$ 
such that $\rho_{\leq |h|} = h$ and for $j >
|h|$, if $\rho_j \in V_i$, then $\rho_{j+1} = \sigma_i(\rho_{\le j})$.  Notice
that when~$h$ is a vertex, then this corresponds to starting the game at that
vertex. When $\calG$ is clear from the context we shall omit it and write
simply $\outcome_h(\sigma_\Agt)$.
If $S_i$ is a set of strategies for
\player{i}, we write $\outcome_h(S_i)$ for $\{\rho \st \exists \sigma_i \in S_i,
\sigma_{-i}\in \stratset_{-i} : \outcome_h(\sigma_i,\sigma_{-i}) = \rho
\}$. Here, $\outcome_h(S_i)$ is the set of outcomes that are \newdef{compatible
with} $S_i$.
All notations for outcomes are lifted to histories in the
obvious way. For a strategy profile $\sigma_\Agt \in \stratset$, we write
$\history_h(\sigma_\Agt)$ for the set $\{ \rho_{\le j} \st \rho \in
\outcome_h(\sigma_\Agt), j \ge |h| \}$.

Consider two strategies $\sigma$ and $\tau$ for \player{i}, and a history~$h$.
We denote by $\sigma\switch{h}{\tau}$ the strategy that follows
strategy~$\sigma$ and \newdef{shifts} to~$\tau$ at history~$h$.

Formally, given a history $h'$ such that $\last(h') \in V_i$:
\[
	\sigma\switch{h}{\tau}(h') \defeq
		\begin{cases}
			\tau(h^{-1} \cdot h') & \text{if } h
				\prefix h' \\
			\sigma(h') & \text{otherwise;}
		\end{cases}
\]

We now formally define dominance and admissibility.  We recall the intuition: a
player's strategy~$\sigma$ is dominated by another strategy $\sigma'$ of his if
$\sigma'$ yields a payoff which is as good as that of~$\sigma$ against all
strategies for the other players, and is strictly better against some of them.
A strategy is admissible if no other strategy dominates it. More formally, we
have:

\subparagraph{Dominance.}
	A strategy $\sigma_i \in \Sigma_i$ \newdef{very weakly dominates}
	strategy $\sigma_i' \in \Sigma_i$, written $\sigma_i \vwdom{}
	\sigma'_i$, if
	\(
		\forall \sigma_{-i} \in \Sigma_{-i},
		\pfunction{i}\big(\outcome_\vinit(\sigma'_i,\sigma_{-i})\big)
		\le
		\pfunction{i}\big(\outcome_\vinit(\sigma_i,\sigma_{-i})\big).
	\)
	Strategy $\sigma_i$ \newdef{weakly dominates} strategy $\sigma'_i$,
	written $\sigma \wdom{} \sigma'$, if $\sigma \vwdom{} \sigma'$ and
	$\neg(\sigma' \vwdom{} \sigma)$.  A strategy $\sigma \in \Sigma_i$ is weakly
	dominated if there exists $\sigma' \in \Sigma_i$ such that $\sigma'
	\wdom{} \sigma$.  A strategy that is not weakly dominated is
	\newdef{admissible}. We denote by $\admstratset_i(\calG)$ the set of all
	admissible strategies for \player{i} in $\calG$.

Our characterizations and algorithms are based on the notions of
\newdef{cooperative} and \newdef{antagonistic values} of a history. The
antagonistic value, denoted~$\aVal_i(\calG,h)$, is the maximum payoff that
\player{i} can secure from~$h$ in the worst case, \textit{i.e.}
against all strategies of other players.
The cooperative value, denoted~$\cVal_i(\calG,h)$, is the
best value \player{i} can achieve from~$h$ with the help of other players.
We also define a third type of value: the \emph{antagonistic-cooperative value},
denoted~$\acVal_i(\calG,h)$, which is the maximum value \player{i} can achieve
in $\calG$ with the help of other players while guaranteeing the antagonistic
value of the current history $h$.  Formal definitions follow.

\subparagraph{Antagonistic \& Cooperative Values.}
	The \newdef{antagonistic value of a strategy} and the
	\newdef{cooperative value of a strategy} $\sigma_i$ of \player{i}
	in $\calG$, for a history $h$ are\\
	\begin{minipage}[b]{0.49\linewidth}
	\[
		\aVal_i(\calG,h,\sigma_i) \defeq \inf_{ \tau \in
		\Sigma_{-i}} \pfunction{i}\big( 
		\outcome_{h}(\sigma_i,\tau) \big);
	\]
	\end{minipage}
	\hfill
	\begin{minipage}[b]{0.49\linewidth}
	\[
		\cVal_i(\calG,h,\sigma_i) \defeq \sup_{\tau \in
		\Sigma_{-i}} \pfunction{i}\big(
		\outcome_{h}(\sigma_i,\tau) \big).
	\]
	\end{minipage}
	The \newdef{antagonistic value of a
	history}~$h$ for \player{i}, and the \newdef{cooperative value of a
	history}~$h$ for \player{i} are defined as
	\(
		\aVal_i(\calG,h) \defeq \sup_{\sigma_i \in \Sigma_i}
		\aVal_i(\calG,h,\sigma_i),
	\)
	and
	\(
		\cVal_i(\calG,h) \defeq \sup_{\sigma_i \in \Sigma_i}
		\cVal_i(\calG,h,\sigma_i),
	\)
	respectively.
	Finally, the \newdef{antagonistic-cooperative value of a history}~$h$
	for \player{i} is
	\[
		\acVal_i(\calG,h) \defeq \sup \{\cVal_i(\calG, h,\sigma_i) \mid
		\sigma_i \in \Sigma_i, \aVal_i(\calG, h,\sigma_i) \geq
		\aVal_i(\calG, h)\}.
	\]  
	We omit $\calG$ when it is clear from the context. 

Observe that $\aVal_i(h)$ of a history is the value of a zero-sum 
two-player game where \player{i} is playing against players~$-i$;
while~$\cVal_i(h)$ is the value in a one-player game, when all
players play together. $\acVal_i(h)$ is a new notion which is
the supremum of the values \player{i} can obtain when he plays
\emph{worst-case optimal strategies}. A strategy $\sigma_i \in \stratset_i$ is
said to be worst-case optimal for \player{i} at history $h$ if
$\aVal_i(h,\sigma_i) = \aVal_i(h)$; it is said to be \emph{cooperatively
optimal} for him at history $h$ if $\cVal_i(h,\sigma_i) = \cVal_i(h)$.
Observe that $\acVal_i(h) = -\infty$ if there are no worst-case optimal
strategies from~$h$.

\begin{example}[Local conditions are not sufficient]
The game in Fig.~\ref{fig:relB-insufficient} shows
that admissibility requires one to consider the values of the histories both in
the past and in the future of the current history. This shows that a local condition
cannot capture admissibility. 
In fact, consider strategy~$\sigma_1$ of \player{1} (who controls all square vertices) that takes
the edges $(s_1,s_2), (s_4,s_6)$. If the game starts at~$s_2$, $\sigma_1$ is
admissible, since the choice $(s_4,s_5)$ could yield a payoff of~$2$ which is worse
than any payoff from~$s_6$. Indeed, we have that $\aVal_1(s_5) < \aVal_1(s_6)$.
However, when the game starts at~$s_1$, $\sigma_1$ is weakly dominated by the
strategy that chooses $(s_1,s_3)$ since the worst payoff in the latter case
is~$5$.  In fact, when a strategy takes the edge $(s_1,s_2)$, the antagonistic
value decreases from~$\aVal_1(s_1)=5$ to~$\aVal_1(s_2)=3$; so to be admissible,
it should have a better cooperative value than~$5$,
which is not the case if $(s_4,s_6)$ is taken. The strategy taking
$(s_1,s_2), (s_4,s_5)$ is admissible. Indeed, in one outcome, the
payoff is~$9$, which is greater than~$5$ as required. Thus, an admissible
strategy from~$s_1$ either goes to~$s_3$, or goes to $s_2$ but commits to
taking~$(s_4,s_5)$ later.
\end{example}

\begin{figure}
\begin{minipage}[b]{0.5\linewidth}
  \centering
  \resizebox{0.66\textwidth}{!}{%
  \begin{tikzpicture}
    \tikzstyle{every state}=[node distance=1.4cm,minimum size=7mm, inner sep=1pt];
    \node[state,rectangle] at (0,0) (s1){$s_1$};
    \node[state,right of=s1] (s2) {$s_2$};
    \node[state,node distance=1.4cm, above right of=s1] (s3) {$s_3$};
    \node[state, rectangle, right of=s2] (s4) {$s_4$};
    \node[state,node distance=1.4cm, above right of=s4] (s5) {$s_5$};
    \node[state,node distance=1.4cm, right of=s4] (s6) {$s_6$};
    \node[inner sep=1pt,above of=s5](s7) {2};
    \node[inner sep=1pt,above right of=s5](s8) {9};
    \node[inner sep=1pt,right of=s6](s9) {4};
    \node[inner sep=1pt,above right of=s6](s10) {3};
    \node[inner sep=1pt,right of=s3](s11) {5};
    \node[inner sep=1pt,above right of=s3](s12) {10};

    \path[-latex'] (s1) edge (s2)
      (s1) edge (s3)
      (s2) edge (s4)
      (s4) edge (s5)
      (s4) edge (s6)
      (s5) edge (s7)
      (s5) edge (s8)
      (s6) edge (s9)
      (s6) edge (s10)
      (s3) edge (s11)
      (s3) edge (s12);
  \end{tikzpicture}
  }
  \caption{Example game where local conditions fail to capture
  admissibility.}
  \label{fig:relB-insufficient}
\end{minipage}
\hfill
\begin{minipage}[b]{0.45\linewidth}
  \centering
  \resizebox{0.75\textwidth}{!}{%
  \begin{tikzpicture}
    \tikzstyle{every state}=[node distance=2cm,minimum size=7mm, inner sep=1pt];
    \node[state,rectangle] at (0,0) (s1){$s_1$};
    \node[state,right of=s1] (s2) {$s_2$};
    \node[left=1cm of s1] (s3) {1};
    \node[right=1cm of s2] (s4) {2};

    \path
    (s1) edge[bend left] (s2)
    (s1) edge (s3)
    (s2) edge[bend left] node[above]{$0$} (s1)
    (s2) edge (s4) ;
 \end{tikzpicture}
 }
 \caption{Example game with an infinite dominance chain and no admissible
   strategy as witness of their being dominated.}
 \label{fig:infinite-dom-chain}
 \end{minipage}
\end{figure}

We use temporal logic to describe sets of outcomes.  We consider an extension
of standard \LTL{} with inequality conditions on payoffs for each player as
in~\cite{BCHK-acm14}.
The logic, denoted $\ltlpayoff$, extends \LTL{}, and its syntax
is defined as follows.
\[
	\phi::= Q
	\mid \lnot \phi 
	\mid \X \phi
	\mid \G \phi
	\mid \F \phi
	\mid \phi_1 \U \phi_2
	\mid \phi_1 \lor \phi_2
	\mid \phi_1 \land \phi_2
	\mid \pfunction{i} \bowtie v,
\]
where~$Q \in \ap$ is a set of atomic propositions on edges, $\G$ and~$\F$ are
the standard \LTL{} modalities, $\bowtie \in \{\mathrel{\leq},
\mathrel{\geq},\mathrel{<},\mathrel{>}\}$, and~$v \in \mathbb{Q}$.
A formula is interpreted over an outcome~$\rho$ at index~$k$ as follows. We
have, for instance, $(\rho,k) \models Q$ if, and only if, $(\rho_k,\rho_{k+1})$ is
labelled with $Q$. For convenience, we write $\rho \models \phi$ instead of
$(\rho,1)
\models \phi$.  Note that we define our predicates on edges rather than vertices;
this simplifies our presentation.  The semantics of the \LTL{} modalities are
standard; we refer to \textit{e.g.} \cite{BCHK-acm14}. For payoff conditions, we
have
\(
	(\rho,k) \models \pfunction{i} \bowtie v 
	\defiff \pfunction{i}(\rho_{\ge k}) \bowtie v.
\)

\subparagraph{Residual Games.}
Given game~$\calG$, and history~$h$, let us define $\calG_h$ as the
\emph{residual game of~$\calG$ from~$h$} by modifying the initial state
to~$\last(h)$, and the payoff functions to $\pfunction{i}'$ defined as follows.
For all outcomes~$\rho$ that start in $\last(h)$, $\pfunction{i}'(\rho) =
\pfunction{i}(h \Cdot \rho)$, where $h\Cdot \rho = h_{\leq |h|-1} \cdot \rho$.
Notice that the strategy sets of~$\calG$ and~$\calG_h$ are identical, and that
for any~$\sigma_\Agt \in \stratset_\Agt$, we have~$\outcome(\calG_h,\sigma_\Agt)
= \outcome_{\last(h)}(\calG, \sigma_\Agt)$.

\begin{lemma}
  \label{lemma:residual}
  For all~$h' \in \history_{\last(h)}(\calG)$, it holds that
  $\aVal_i(\calG_h,h') = \aVal_i(\calG,h \Cdot h')$, $\acVal_i(\calG_h,h') =
  \acVal_i(\calG,h \Cdot h')$, and~$\cVal_i(\calG_h, h') = \cVal_i(\calG, h\Cdot
  h')$.
\end{lemma}

\section{Existence of Admissible Strategies}\label{sec:existence}

We start this section with two examples of quantitative games with no admissible strategies (for \player{1}). Then we identify a large and natural class of games for which the existence of admissible strategies is guaranteed.

  \begin{figure}[ht]
    \begin{minipage}{0.48\textwidth}
    \centering
    \begin{tikzpicture}
      \tikzstyle{every state}=[rectangle,node distance=2cm,minimum size=5mm, inner sep=1pt];
      \node[state,initial above] at (0,0) (s1){$s_1$};
      \node[state, left of=s1,accepting] (s3) {$s_3$};
      \node[state, right of=s1, node distance=1.3cm] (a) {$a$};
      \path[-latex'] (s1) edge (s3);
      \path[-latex']
      (s1) edge[bend left] (a)
      (a) edge[bend left] (s1)
      (s3) edge[loop left] (s3);      
    \end{tikzpicture}

    Game~$\calA$.
    \end{minipage}
    \begin{minipage}{0.48\textwidth}
    \centering
    \begin{tikzpicture}
      \tikzstyle{every state}=[node distance=2cm,minimum size=5mm, inner sep=1pt];
      \node[state,rectangle,initial above] at (0,0) (s1){$s_1$};
      \node[state,rectangle, left of=s1,accepting] (s3) {$s_3$};
      \node[state, right of=s1] (s2) {$s_2$};
      \node[state,rectangle, above left of=s2, node distance=1cm] (a) {$a$};
      \path[-latex'] (s1) edge (s3);
      \path[-latex'] (s1) edge (s2)
      (s2) edge (a)
      (a) edge (s1)
      (s2) edge[bend left] (s1)
      (s3) edge[loop left] (s3);      
    \end{tikzpicture}

    Game~$\calG$.
    \end{minipage}
    \caption{Two games in which Player~1 has no admissible strategy.}
    \label{fig:no-adm}
  \end{figure}
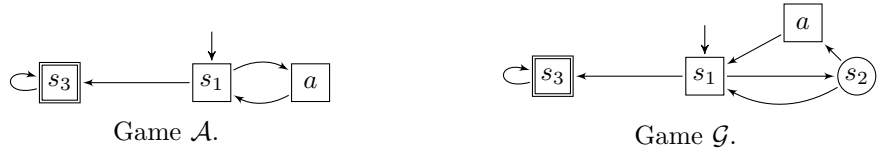

  Consider the games~$\calA$ and~$\calG$ in Fig.~\ref{fig:no-adm}.  Starting at~$s_1$, the payoff of \player{1}, in the two games is defined as follows: an outcome that does not visit~$s_3$ has a payoff equal to~$0$, otherwise, the payoff is equal to the number of times vertex~$a$ appears in the outcome.
The lemma below states that \player{1} does not have admissible strategies in those two games. 
  We sketch the proof idea.
	
	  Consider first the one-player game $\mathcal{A}$.  The antagonistic
	  value at vertex $s_1$ is $\infty$.  Any strategy which never visits
	  $s_3$ is weakly dominated by strategies that visit $a$ at least
	  once (i.e. with outcome $(s_1as_1)^{+}s_3^\omega$).  Furthermore, a
	  strategy which does visit $s_3$ and $k$ times $a$ is weakly dominated
	  by any strategy that visits $a$ at least $k+1$ times and then goes to
	  $s_3$.
	 
	  The idea is similar for $\calG$ where the cooperative value at $s_1$
	  is $\infty$.  Every strategy which does not allow outcomes visiting
	  $s_3$ are weakly dominated by those that attempt to visit $a$ by
	  visiting $s_2$ at least once (as from $s_2$, the other player can
	  cooperate and visit $a$), and then go to $s_3$. Moreover, it is always
	  possible to attempt to visit $a$ once more before going to $s_3$, thus
	  any strategy which eventually goes to $s_3$ is also weakly dominated.
	  
  \begin{lemma}\label{lem:sometimes-no-adm}
	  \Player{1} does not have admissible strategies in games~$\calG$
	  and~$\calA$.
  \end{lemma}
		
In the two examples above, either the $\aVal$ or the $\cVal$ (which are both
equal to $\infty$) are not achievable. This is not a coincidence.
We now show that all the games that admit witnessing strategies for those values
are guaranteed to have admissible strategies.

  \subparagraph{Games with strategies witnessing $\aVal$ and $\cVal$.}
 A game is {\em well-formed} whenever it admits witnessing strategies for
 $\aVal$ and $\cVal$, \ie it satisfies:
  
\newcounter{asscounter}
  \begin{enumerate}
  \item \label{ass:optimal-strats}
    For all $i
    \in \Agt$, and $h \in \history_\vinit(\calG)$,
    \(
    \exists \sigma_i \in \stratset_i, \aVal_i(h,\sigma_i)
    = \aVal_i(h).
    \)
  \item \label{ass:optimal-coop-strats}
    For all $i
    \in \Agt$, and $h \in \history_\vinit(\calG)$,
    \(
    \exists \sigma_i \in \stratset_i, \cVal_i(h,\sigma_i)
    = \cVal_i(h).
    \)
    \setcounter{asscounter}{\theenumi}     
  \end{enumerate}

  These conditions will also be referred as Assumption~\ref{ass:optimal-strats}
  and~\ref{ass:optimal-coop-strats}.

We now establish the existence of admissible strategies for all well-formed games.

\begin{theorem}\label{thm:existence}
	In all well-formed games all
	players have admissible strategies.
\end{theorem}
The result follows from Lemmas.~\ref{lem:sco-are-adm}
and~\ref{lem:strongcoopop-exist} below: the proof consists in showing
that a particular type of admissible strategies, called {\em strongly
cooperative-optimal}, always exists.  Usually, those strategies are only a strict
subset of the admissible strategies available to a player. Nevertheless, they are
peculiar as they are guaranteed to exist.

\begin{definition}
	A strategy~$\sigma_i$ is \newdef{strongly
	cooperative-optimal (SCO)} if for all $h \in \history_\vinit(\sigma_i)$,
	if $\cVal_i(h) > \aVal_i(h)$ then $\cVal_i(h,\sigma_i) = \cVal_i(h)$,
	and if $\aVal_i(h) = \cVal_i(h)$ then $\aVal_i(h,\sigma_i) =
	\aVal_i(h)$.
\end{definition}
%
Strongly cooperative-optimal strategies are admissible
because their cooperative values are always maximal, and moreover, if a payoff
better than the antagonistic value cannot be achieved ($\aVal_i(h) =
\cVal_i(h)$), then they are worst-case optimal.
Any strategy which obtains a better payoff than a SCO strategy 
against some adversary will obtain a worse payoff against another one.

\begin{lemma}\label{lem:sco-are-adm}
	All strongly cooperative-optimal strategies are admissible.
\end{lemma}
\begin{proof}
	Let $\sigma_i$ be a strongly cooperative-optimal strategy for
	\player{i}.  
        Assume towards a contradiction that some $\sigma_i'$ weakly dominates~$\sigma_i$.
	Let~$h$ be any minimal history compatible with $\sigma_i$ such
	that~$\sigma_i(h) \neq \sigma_i'(h)$.  

	If~$\aVal_i(h) < \cVal_i(h)$, then since $\last(h)$ is controlled by \player{i}, $\aVal_i(h \sigma_i'(h)) \le \aVal_i(h) < \cVal_i(h)$,
        and since $\sigma_i$ is strongly cooperative optimal $\cVal_i(h \sigma_i(h),\sigma_i) = \cVal_i(h)$.
        Therefore, as the histories~$h\sigma_i(h)$ and~$h\sigma_i'(h)$ are distinct, there is a strategy~$\tau \in \Sigma_{-i}$
        such that
	$\pfunction{i}(\outcome_{h\sigma_i(h)}(\sigma_i, \tau))= \cVal_i(h)  > \aVal_i(h) \ge \pfunction{i}(\outcome_{h\sigma_i'(h)}(\sigma_i', \tau))$.
        This contradicts that~$\sigma_i'$ weakly dominates~$\sigma_i$.

	Otherwise~$\aVal_i(h) = \cVal_i(h)$, then since $\sigma_i$ is strongly cooperative optimal, for all~$\tau \in \Sigma_{-i}$, $\pfunction{i}(\outcome_h(\sigma_i,\tau)) = \cVal_i(h)$ and $\pfunction{i}(\outcome_h(\sigma_i',\tau)) \leq \cVal_i(h)$.
	It follows that
        no outcome of~$\sigma_i'$ obtains a better payoff than~$\sigma_i$.
        We thus obtain a contradiction.
\end{proof}

By Lem.~\ref{lem:sco-are-adm}, to prove the existence of
admissible strategies, it suffices to prove the existence of strongly
cooperative-optimal strategies.  We actually give a constructive
proof.

\begin{lemma}\label{lem:strongcoopop-exist}
 	In all well-formed games all players have SCO strategies.
\end{lemma}
  Let us describe the idea of the construction.
	Consider any player~$i$.  We define the strategy~$\sigma$ of player~$i$ as follows.  For any
  history~$h$, if~$\aVal_i(h) = \cVal_i(h)$, then $\sigma$ plays a worst-case
	optimal strategy from~$h$, say~$\sigma_h^\wco$.
	Otherwise, we define $\sigma$ starting from an outcome~$\rho_h$ with~$\pfunction{i} = \cVal_i(h)$, and
	we define $\sigma$ is such a way that it follows~$\rho_h$. In this case, whenever another player deviates from~$\rho_h$, say, at
  history~$h'$, we reevaluate how to play according to whether~$\aVal_i(h') <
  \cVal_i(h')$ or~$\aVal_i(h')= \cVal_i(h')$.
	Here, the existence of~$\sigma_h^\wco$ and that of~$\rho_h$ are guaranteed
	by the fact that the game is well-formed.

In subsequent sections, we consider SCO strategies in residual games~$\calG_h$,
so let us note that these games satisfy the required assumptions if
$\calG$ does,
which follows from Lem.~\ref{lemma:residual}.
\begin{lemma}\label{lem:residual-admissible}
	For any well-formed game~$\calG$,
	for all histories $h \in \history_\vinit$, the residual game $\calG_h$ is also well-formed.
\end{lemma}

We end this section with
an interesting observation: an infinite weak dominance chain is not
necessarily dominated by an admissible strategy, as shown in the next example.
The
reader should contrast the example with the fact that in the Boolean case all
dominated strategies are dominated by an admissible
strategy~\cite[Thm.~11]{berwanger07}.

\begin{example}[Non-dominated weak dominance chains]\label{exa:infinite-chain}
There are quantitative games that have
infinite dominance chains and no ``maximal''
admissible strategy weakly dominating them. 
Consider the game depicted in
Fig.~\ref{fig:infinite-dom-chain}. 
Denote by $\sigma^k$ the strategy of \player{1} (controlling square vertices)
which consists in moving from $s_1$ to $s_2$ exactly $k$ times, and then going
left (unless payoff of~$2$ was reached in the meantime).  Then for all $k \in
\mathbb{N}$, $\sigma^k$ is weakly dominated by $\sigma^{k+1}$ because if the
adversary decides to move right from~$s_2$ at the~$(k+1)$-th step,
$\sigma^{k+1}$ performs better than~$\sigma^k$, and otherwise they yield
identical outcomes.  It follows that all strategies $\sigma^k$ for~$k\geq 0$,
are dominated.  Here, the only admissible strategy~$\sigma^\infty$ consists in
looping in the cycle forever, which does not dominate any~$\sigma^k$  since if
the adversary always moves left from~$s_2$, then~$\sigma^\infty$ yields less
than~$\sigma^k$.
\end{example}

\begin{remark}
Above, we have defined strongly cooperative-optimal strategies that favour
cooperation whenever it can have an added value. We have established that those
strategies are always admissible. There are other classes of strategies that are
always admissible, and we define another interesting class here. A strategy
$\sigma_i$ is a {\em worst-case cooperative optimal strategy},   if for all $h
\in \history_\vinit(\sigma_i)$:
		$\aVal_i(h,\sigma_i)=\aVal_i(h)$, and
		$\cVal_i(h,\sigma_i)=\acVal_i(h)$.
\end{remark}
 So those strategies ensure the worst-case value at all times
 and leave open the best cooperation possible under that worst-case guarantee.

\begin{lemma}
All worst-case cooperative optimal strategies strategies are admissible.
\end{lemma}
However, some well-formed games do not have worst-case cooperative
optimal strategies.

\section{Value-based Characterization of Admissible Strategies}
\label{sec:val-characterization}
We present our main result, which is, a value-based
characterization of admissible strategies.

For any game~$\calG$, and player~$i$, let us define the following property,
denoted $\propadm{h}{\sigma}$, for a given strategy~$\sigma \in
\stratset_i(\calG)$ and history~$h$:
\begin{align}
  \label{eqn:adm1}
  \cVal_i(h,\sigma) > \aVal_i(h)\\
  \lor \quad\aVal_i(h,\sigma) = \cVal_i(h,\sigma) =
  \aVal_i(h) = \acVal_i(h),
  \label{eqn:adm2}
\end{align}

Intuitively, we will show that a strategy is admissible if at all  histories, 
either the strategy promises
a cooperative value greater than the antagonistic value at the current vertex, or
a higher cooperative value cannot be obtained without risking a lower antagonistic value
(\textit{i.e.} $\aVal_i(h)= \acVal_i(h)$) and the strategy is worst-case optimal.

It turns out that requiring this property at all histories ending in a player's vertices
characterize admissible strategies.
We state our result in the following theorem.
\begin{theorem} \label{thm:lhd-admissibles}
	Under Assumption~\ref{ass:optimal-strats}, for any game~$\calG$, \player{i}, and $\sigma_i \in
	\stratset_i(\calG)$, 
     $\sigma_i$ is admissible if, and only if, for
	all~$h \in \history_{\vinit}(\calG,\sigma_i)$ with $\last(h) \in V_i$,
	$\propadm{h}{\sigma_i}$ holds.
\end{theorem}

It will be useful to consider the negation of~$\propadm{h}{\sigma}$, which we
simplify as follows:
\begin{lemma}\label{lem:negated-propadam}
	For all histories~$h$ and strategy~$\sigma$, the negation
	of~$\propadm{h}{\sigma}$ is equivalent to
	\begin{align}
		\label{eqn:stlhd1}
		\cVal_i(h,\sigma) \leq \aVal_i(h) \land \aVal_i(h,\sigma) <
			\aVal_i(h)\\
		\label{eqn:stlhd2}
		\lor\quad \cVal_i(h,\sigma) = \aVal_i(h,\sigma) = \aVal_i(h)
			\land \acVal_i(h) > \aVal_i(h).
	\end{align}
\end{lemma}

\begin{proof}[Proof of Thm.~\ref{thm:lhd-admissibles}]
	\fbox{$\Rightarrow$} We prove the contrapositive.  Assume that~$\exists
	h \in \history_\vinit(\calG,\sigma_i)$, $\last(h) \in V_i$ and $\lnot
	\propadm{h}{\sigma_i}$.  Then by Lem.~\ref{lem:negated-propadam},
	either \eqref{eqn:stlhd1} or \eqref{eqn:stlhd2} holds for $(h,
	\sigma_i)$.
  
	Assume~\eqref{eqn:stlhd1} holds for $(h, \sigma_i)$. 
    By Assumption~\ref{ass:optimal-strats}, there exists a \emph{worst-case optimal} strategy
    $\sigma^{\wco}_h$ from~$h$,
    with $\aVal_i(h, \sigma^{\wco}_h) = \aVal_i(h)$.
    Define~$\sigma'_i \defeq \sigma_i[h
	\leftarrow \sigma^{\wco}_h]$.  We claim that~$\sigma'_i$ weakly
	dominates~$\sigma_i$.  In fact, for any~$\sigma_{-i} \in
	\stratset_{-i}(\calG)$ with~$h \not \in
	\history_{\vinit}(\calG,\sigma_{-i})$, we
	have $\outcome_{\vinit}(\calG,\sigma_i,\sigma_{-i}) =
	\outcome_{\vinit}(\calG,\sigma'_i,\sigma_{-i})$.  For any~$\sigma_{-i}
	\in \stratset_{-i}(\calG)$ compatible with~$h$,
    both outcomes go through~$h$.  
	By definition of $\sigma'_i$,
	$\outcome_{\vinit}(\calG,\sigma'_i,\sigma_{-i}) = h_{\le |h|-1} \cdot
	\outcome_{h}(\calG,\sigma^{\wco}_h,\sigma_{-i})$.  Therefore, we have
	that $\pfunction{i}(\outcome_{\vinit}(\calG,\sigma'_i,\sigma_{-i}))
 	\ge \aVal_i(h)$ by definition of $\sigma^{\wco}_h$. The latter is
	greater than $\cVal_i(h, \sigma_i)$ from \eqref{eqn:stlhd1}, so
	greater than
	$\pfunction{i}(\outcome_{\vinit}(\calG,\sigma_i,\sigma_{-i}))$ by
	definition of $\cVal_i(\cdot)$.
	Thus, $\sigma'_i$ very weakly dominates~$\sigma_i$.  Since by
	assumption, $\aVal_i(h,\sigma_i) < \aVal_i(h)$, and $h$ is compatible
	with $\sigma_i$, there is a strategy $\sigma_{-i} \in \Sigma_{-i}$
	such that
$h \prefix \outcome_\vinit(\calG, \sigma_i,\sigma_{-i})$ and
	$\pfunction{i}(\outcome_\vinit(\calG,\sigma_i,\sigma_{-i})) <$
	$\aVal_i(h)$.  As shown before,
	\(
		\aVal_i(h) \le
		\pfunction{i}(\outcome_{\vinit}(\calG,\sigma'_i,\sigma_{-i})).
	\)
        Hence, $\sigma'_i$ weakly dominates~$\sigma_i$.

	Assume now that~\eqref{eqn:stlhd2} holds. 
    Consider~$\epsilon>0$ small enough so that $\acVal_i(h) > \aVal_i(h)+\epsilon$.
    By definition of~$\acVal_i(h)$, there exists a strategy~$\tau_i \in \Sigma_i$
    such that $\cVal_i( h, \tau_i) \geq \aVal_i(h) + \epsilon$,
    and moreover $\aVal_i(h, \tau_i) \geq \aVal_i(h)$.
    Consider~$\tau_{-i} \in \Sigma_{-i}$ compatible with~$h$ such that
    \(
      \pfunction{i}(\outcome_h(\calG, h, (\tau_i,\tau_{-i}))) \geq \cVal_i(h,\tau_i)-\frac{\epsilon}{2}
      \geq \aVal_i(h) + \frac{\epsilon}{2} > \aVal_i( h).
    \)
    Note that such a $\tau_{-i}$ exists by definition of~$\cVal_i( h, \tau_i)$.
    It follows that $\sigma_i[h \leftarrow \tau_i]$ weakly dominates~$\sigma_i$.
    In fact, the outcomes are identical for any outcome not compatible with~$h$.
    For any~$\sigma_{-i}$ compatible with~$h$,
	we have $\pfunction{i}(\outcome_{\vinit}(\calG,\sigma_i,\sigma_{-i})) =
	\aVal_i(h)$ by~\eqref{eqn:stlhd2}.  Moreover,
	$\pfunction{i}(\outcome_{\vinit}(\calG,\sigma'_i,\sigma_{-i})) \geq
	\aVal_i(h)$ since at~$h$
    we have that $\aVal_i( h, \tau_i)\geq \aVal_i(h)$;
    thus $\aVal_i( h, \sigma_i') \geq \aVal_i(h)$.
    Furthermore, we have $\pfunction{i}(\outcome_{\vinit}(\calG, \sigma_i',\tau_{-i})) > \aVal_i(h)
    \geq \pfunction{i}(\outcome_{\vinit}(\calG, \sigma_i, \tau_{-i}))$.

	\fbox{$\Leftarrow$} Assume that for all $h \in \history_\vinit(\calG,\sigma_i)$
	with $\last(h)\in V_i$, we have $\propadm{h}{\sigma_i}$, and that $\sigma_i$
	is weakly dominated by some strategy~$\sigma'_i$.  
        We will show a contradiction.

	Let $\sigma_{-i}$ be a strategy in $\stratset_{-i}(\calG)$ and $\rho =
	\outcome_{\vinit}(\calG,\sigma_i,\sigma_{-i})$ and $\rho' =
	\outcome_{\vinit}(\calG,\sigma'_i,\sigma_{-i})$.  If $\rho = \rho'$ then
	$\pfunction{i}(\rho') \le \pfunction{i}(\rho)$ and otherwise let $j$ be
	the first index where they differ, and $h = \rho_{\le j - 1} =
	\rho'_{\le j - 1}$.  We have that $h$ is compatible with both
	strategies, $\last(h) \in V_i$ and $\sigma_i(h) \ne \sigma_{i}'(h)$.

	If~\eqref{eqn:adm1} holds, that is,~$\cVal_i(h,\sigma_i) > \aVal_i(h)$,
    consider~$\epsilon>0$ such that ~$\cVal_i(h,\sigma_i) > \aVal_i(h)+\epsilon$,
	and a strategy~$\sigma'_{-i} \in \stratset_{-i}$ which
	ensures that
	$\pfunction{i}(\outcome_{h\sigma_i(h)}(\calG, \sigma_i,\sigma'_{-i})) \geq
	\cVal_i(h,\sigma_i)-\frac{\epsilon}{2}$, and
	$\pfunction{i}(\outcome_{h\sigma'_i(h)}(\calG, \sigma'_i,\sigma'_{-i})) \leq
	\aVal_i(h,\sigma'_i)+\frac{\epsilon}{2}$.  Such a strategy
	profile~$\sigma'_{-i}$ exists since $h\sigma_i'(h)$ and~$h\sigma_i(h)$
	are distinct, and since~$\last(h) \in V_i$. The latter also implies that
	$\aVal_i(h) \geq \aVal_i(h,\sigma_i')$.
	It thus follows that
    \[
	    \pfunction{i}(\outcome_{\vinit}(\calG, \sigma_i,\sigma_{-i}[h
	    \leftarrow \sigma_{-i}'])) >
	    \pfunction{i}(\outcome_{\vinit}(\calG, \sigma_i', \sigma_{-i}[h
	    \leftarrow \sigma_{-i}']))
    \]
	contradicting the fact that $\sigma_i'$ weakly dominates $\sigma_i$.

	Therefore~\eqref{eqn:adm2} must hold, and $\acVal_i(h) = \aVal_i(h)$.  
    If there exists~$j \geq |h|$ such that $\aVal_i(\rho'_{\leq j}) < \aVal_i(h)$,
    then there exists~$\epsilon>0$ and a strategy profile $\sigma'_{-i}\in \stratset_{-i}$
	compatible with $h$ which ensures
	that $\pfunction{i}(\outcome_{\vinit}(\calG,\sigma_i',\sigma'_{-i})) \leq \aVal_i(\rho_{\leq j}')+\epsilon < \aVal_i(h) \leq \pfunction{i}(\outcome_{\vinit}(\calG,\sigma_i,\sigma'_{-i}))$. 
        This contradicts  $\sigma'$ weakly dominating~$\sigma$.
        Hence for all $j \geq |h|$, $\aVal_i(\rho'_{\leq j}) \ge \aVal_i(h)$. 
        Now, observe that $\pfunction{i}(\rho') \le \acVal_i(h)$. In fact, one can construct a strategy~$\tau$, which, from~$h$ follows~$\rho'$,
        and in case another player does not respect~$\rho$, switches to a worst-case optimal strategy ensuring~$\aVal_i(\rho'_{\leq j})\geq \aVal_i(h)$.
        It follows that~$\pfunction{i}(\rho') \leq \cVal_i(h,\tau) \leq \acVal_i(h)$.
        Furthermore, by~\eqref{eqn:adm2}, $\pfunction{i}(\rho) \geq \acVal_i(h) = \aVal_i(h,\sigma_i)$,
        so $\pfunction{i}(\rho') \leq \pfunction{i}(\rho)$.
        This being true for all strategies of $\stratset_{-i}$ proves that~$\sigma_i$ very weakly dominates~$\sigma'_i$ and contradicts that $\sigma'_i$ weakly dominates~$\sigma_i$.
\end{proof}

\section{Characterization of the Outcomes of Admissible Strategies}
\label{sec:ltl-characterization}
Observe that the characterization of Thm.~\ref{thm:lhd-admissibles} does not
immediately yield an effective representation of the \emph{set} of admissible
strategies.  In order to reason about the possible behaviors observable in a
game under admissible strategies we are interested in describing the set of
outcomes that can be observed when all players play admissible strategies.  In
this section, for each player, we give a linear temporal logic description of
the outcomes that are each compatible with at least one admissible strategy.

Note that our main goal is to obtain such a characterization in full
generality, for all well-formed games
so we defer
computability considerations to the next section. We will then see how
the three types of values can be computed at all histories.

Let us fix a game~$\calG$, and \player{i}.
We present the intuition of the characterization. If an outcome~$\rho$ is
compatible with an admissible strategy, say~$\sigma_i$, then all prefixes~$h$
with $\last(h) \in V_i$ must satisfy~\eqref{eqn:adm1} or~\eqref{eqn:adm2}.
Given~$h$, if~\eqref{eqn:adm1} holds, then two things can happen.  Either
$\pfunction{i}(\rho) > \aVal_i(h)$, and thus $\rho$ witnesses
$\cVal_i(h,\sigma_i) > \aVal_i(h)$, or this is not the case but there is another
outcome~$\rho'$---compatible with $\sigma_i$---extending~$h$
with~$\pfunction{i}(\rho')>\aVal_i(h)$.  Notice how the longest common prefix
of~$\rho$ and~$\rho'$ ends always with a vertex in~$V_{-i}$ since both outcomes are
compatible with~$\sigma_i$.  If~\eqref{eqn:adm2} holds at~$h$, then, in
particular, $\pfunction{i}(\rho) = \aVal_i(h)$ and, moreover, $\aVal_i$ remains
constant at all prefixes of~$\rho$ extending~$h$.  The last observation simply
follows from $\aVal_i(h,\sigma_i) = \aVal_i(h)$ which is implied
by~\eqref{eqn:adm2}.


\subparagraph{Extended $\ltlpayoff$.}
Let $\mathbf{aValues}_i = \{\aVal_i(h) \st h \text{ is a history}\}$ be the set
of antagonistic values of \player{i}. We will now define atomic propositions
attached to edges of a game. Formally, we have a \newdef{labelling function}
$\lambda : E \to \pow(\ap)$ which assigns to every edge a set of propositions
from $\ap$. The set $\ap$ includes the proposition $\ownProp$ whose truth value,
for every edge $e =(u,v)$, is determined as follows:
$\ownProp \in \labeling(u,v) \defiff u \in V_i$.

We consider $\ltlpayoff$ with atomic propositions as defined above and
additional propositions $\aValProp_q$, $\acValProp_q$, and $\gAlt_q$ defined for
all $q \in \mathbf{aValues}_i$. The semantics of these are straightforward: for
an outcome $\rho$ and $k \in \mathbb{N}_{>0}$ we have
\begin{align*}
   (\rho,k) \models \aValProp_q  &\defiff \aVal_i(\rho_{\le k}) = q,\\
   (\rho,k) \models \acValProp_q  &\defiff \acVal_i(\rho_{\le k}) = q, \text{ and}\\
   (\rho,k) \models \gAlt_q  &\defiff  \rho_k \in V_{-i} \land \exists v' \neq
	 \rho_k
   : (\rho_{k},v') \in E \land \cVal_i(\rho_{\le k} \cdot v') > q,
\end{align*}
with the convention that, when $k$ is omitted, we assume it is $1$.

As mentioned earlier, we  consider two
cases depending on whether~\eqref{eqn:adm1} or~\eqref{eqn:adm2} hold. Thus, let
us define the corresponding two sub-formulas:
\[ \textstyle
\phi_{\ref{eqn:adm1}}
	\defeq \bigvee_{q \in \mathsf{aValues}_i}
	\left(\aValProp_q \land \left(\pfunction{i} > q \lor \F (\gAlt_q)
	\right)\right), \text{ and}
\]
\[ \textstyle
	\phi_{\ref{eqn:adm2}} \defeq \bigvee_{q \in \mathsf{aValues}_i} \left(
	\aValProp_q \land \acValProp_q \land \pfunction{i} = q \land \G\left(
	\aValProp_q\right) \right).
\]
We define the following formula which will be shown to capture the outcomes of
admissible strategies:
\(
	\Phi_{\mathsf{adm}}^i \defeq \G \left( \lnot \ownProp \lor
	\phi_{\ref{eqn:adm1}} \lor \phi_{\ref{eqn:adm2}} \right).
\)

\begin{theorem}
  \label{thm:ltl-characterization}
	For any well-formed game~$\calG$, outcome $\rho$ satisfies $\Phi_{\mathsf{adm}}^i$ if, and only if, it
	is compatible with an admissible strategy for \player{i}.
\end{theorem}
	We give the idea of the proof.
	For any outcome $\rho$ compatible with an admissible strategy $\sigma_i$ for
	\player{i}. We show that for any prefix $h$ of $\rho$ with
	$\last(h) \in V_i$, $(\rho,|h|)$ satisfies either $\phi_{\ref{eqn:adm1}}$
	or $\phi_{\ref{eqn:adm2}}$. In fact, by
	Thm.~\ref{thm:lhd-admissibles}, either~\eqref{eqn:adm1} or~\eqref{eqn:adm2}
	hold, and we show that these correspond to $\phi_{\ref{eqn:adm1}}$
	and~$\phi_{\ref{eqn:adm2}}$.
	
	Conversely, for any $\rho$ satisfying
	$\Phi_{\textsf{adm}}^i$, we construct an admissible strategy~$\sigma_i$ for \player{i}
	compatible with~$\rho$. The strategy follows~$\rho$, and in case of deviation,
	it switches immediately either to an SCO---which is guaranteed to
	exist---or to a worst-case optimal
	strategy, depending on whether $\phi_1$ or $\phi_2$ holds at the current
	history. The resulting strategy is proven to be admissible.

\subparagraph{Assuming prefix-independence.} Before concluding this section,
let us consider the consequences of further assuming that our payoff function
is \emph{prefix-independent}. 
\begin{enumerate}
 	\setcounter{enumi}{\value{asscounter}} 	
	\item \label{ass:prefix-indep}
		For
		all $i \in \Agt$, for all outcomes $\rho$, it holds that
		\(
			\forall j \in \mathbb{N}, \pfunction{i}( (\rho_k)_{k \ge
			j}) = \pfunction{i}(\rho).
		\)		
\end{enumerate}

Observe that, under Assumption~\ref{ass:prefix-indep}, the set
$\mathbf{aValues}_i$ can be equivalently defined as $\{\aVal_i(v) \st v \in V\}$
and is thus finite. One can also extend the labelling $\lambda$ and set of
atomic propositions $\ap$ such that, for every edge $e = (u,v)$ and
$q \in \mathbf{aValues}_i$:
\begin{align*}
   \aValProp_q  \in \labeling(u,v) &\defiff \aVal_i(u) = q,\\
   \acValProp_q  \in \labeling(u,v) &\defiff \acVal_i(u) = q, \text{ and}\\
   \gAlt_q  \in \labeling(u,v) &  \defiff  u \in V_{-i} \land \exists v' \neq v : (u,v') \in E 
		\land \cVal_i(v') > q.
\end{align*}
It immediately follows that:
\begin{lemma}\label{lem:prefix-indep}
	Under Assumption~\ref{ass:prefix-indep}, for all $i \in \Agt$, 
	$\Phi_{\mathsf{adm}}^i$ is expressible in $\ltlpayoff$.
\end{lemma}

\section{Applications and Future Works}\label{sec:applications}

In this section, we show how to apply Theorem~\ref{thm:lhd-admissibles}
(value-based characterization of admissible strategies) and
Theorem~\ref{thm:ltl-characterization}
(characterization of the set of outcomes of admissible strategies) to solve
relevant verification and synthesis problems. 

\smallskip
\noindent\textbf{Classical payoff functions.}
So far, we have assumed that games were equipped for each player $i \in \Agt$
with a payoff function. To define payoff functions, we proceed as usual by first
assigning weights to edges of the game graph using  \emph{weight functions} $w_i
: E \to \mathbb{Q}$, one for each player $i \in \Agt$. With the weight function
$w_i$, we associate to each outcome $\rho=\rho_1 \rho_2 \dots \rho_n \dots$, an
infinite sequence of rational values $w_i(\rho)=w_i(\rho_1 \rho_2) w_i(\rho_2
\rho_3) \dots w_i(\rho_n \rho_{n+1}) \dots$, and we aggregate this sequence of
values with measures such as $\inffun$, $\supfun$, $\linffun$, $\lsupfun$,
and mean payoff ($\underline{\mpfun}$ and $\overline{\mpfun}$).
It is well known, see \eg~\cite{cdh10} and~\cite{zp96}, that all the payoff
functions defined above satisfy
Assumptions~\ref{ass:optimal-strats}-\ref{ass:optimal-coop-strats}.
By Theorem~\ref{thm:existence}, we get the following.
\begin{lemma} \label{lem:measures-assumption}
	In games with payoff functions from $\inffun$, $\supfun$, $\linffun$,
	$\lsupfun$, $\underline{\mpfun}$, and $\overline{\mpfun}$, all players
	have admissible strategies.
\end{lemma}

It is also known that, in games defined with the payoff functions considered
here, the antagonistic and cooperative values ($\cVal$ and $\aVal$) are
computable. One can also show that $\acVal$ is computable for prefix-independent
payoff functions. Indeed, this value of a vertex coincides with the $\cVal$
inside the sub-graph induced by the vertices with the optimal antagonistic value.
Furthermore, using a classical
transformation on the game structure, we can guarantee that all payoff functions
above are prefix-independent. 
We thus obtain the following result,
by Lemma~\ref{lem:prefix-indep}.
\begin{lemma} \label{lem:all-computable}
	In games with payoff functions from $\inffun$, $\supfun$, $\linffun$,
	$\lsupfun$, $\underline{\mpfun}$, and $\overline{\mpfun}$, the formulas
	$\Phi^i_{\textsf{adm}}$ for all $i \in \Agt$ are effectively computable,
	finite,	and expressible in $\ltlpayoff{}$.
\end{lemma}

We will now consider several problems of interest which can be solved using the
characterizations that we have developed in the previous sections. All the
results are applicable to the measures concerned by
Lemmas~\ref{lem:measures-assumption} and~\ref{lem:all-computable}.

\smallskip
\noindent\textbf{Deciding the admissibility of a finite memory strategy.} As a
first example, we consider the problem of deciding, given a game structure
$\mathcal{G}$, and a (finite memory) strategy $\sigma_i$ for player $i \in \Agt$
described as a finite state transducer ${\sf M}_i$,  if $\sigma_i$ is admissible
in $\mathcal{G}$. 

To solve this problem, we rely on Theorem~\ref{thm:lhd-admissibles} and proceed
as follows. First, we compute for each vertex $v$ of the game $\mathcal{G}$, the
values $\aVal_i(\calG,v)$, $\cVal_i(\calG,v)$, and $\acVal_i(\calG,v)$.  Second,
we construct the synchronized product between the transducer ${\sf M}_i$ that
defines the strategy $\sigma_i$ and the game $\mathcal{G}$. States in this
product are of the form $(v,m)$ where $v$ is a vertex of $\calG$ and $m$ is a
(memory) state of the transducer ${\sf M}_i$. Third, we compute for each state $(v,m)$
the values $\aVal_i(\calG,(v,m),\sigma_i)$, $\cVal_i(\calG,(v,m),\sigma_i)$, and
$\acVal_i(\calG,(v,m),\sigma_i)$. Finally, we verify that there is no reachable
vertex $(v,m)$ in the product where condition $(1)$ or condition $(2)$ are
falsified.
We then obtain the following theorem:

\begin{theorem}\label{thm:strat-admissible}
Given a game $\calG$ and a finite memory strategy $\sigma_i$ for player $i \in
\Agt$ specified as a finite state transducer ${\sf M}_i$, we can decide if $\sigma_i$
is an admissible strategy for player $i$ in {\sf PTime}~ for measures $\inffun$,
$\supfun$, $\linffun$, $\lsupfun$; in $\NP \cap \coNP$ for 
$\underline{\mpfun}$, and $\overline{\mpfun}$.
\end{theorem}

\smallskip
\noindent\textbf{Model-checking under admissibility.}
We now turn to the following problem. Given a game structure $\calG$ and a
\ltlpayoff{} formula $\phi$, decide if all outcomes of the game that are
compatible with the admissible strategies of all players satisfy~$\phi$, \ie~if
$\bigcap_{i \in \Agt} \outcome(\calG, \admstratset_i(\calG)) \models \phi$. This
problem was introduced in the Boolean setting in~\cite{BRS14} and allows one to
check that a property is induced by the rationality of the players in a game.

\begin{theorem}\label{thm:mc-admissible}
  For all measures $\inffun$, $\supfun$, $\linffun$,
  $\lsupfun$, $\underline{\mpfun}$, $\overline{\mpfun}$, 
  one can decide, given game~$\calG$ and \ltlpayoff{} formula~$\phi$,
  whether $\bigcap_{i \in \Agt} \outcome(\calG,
  \admstratset_i(\calG)) \models \phi$.
\end{theorem}
\begin{proof}[Proof Sketch]
For each player~$i \in \Agt$, consider the formula $\Phi_{\mathsf{adm}}^i$
from Theorem~\ref{thm:ltl-characterization}, which describes
the set $\outcome(\calG, \admstratset_i(\calG))$. The formula is finite and constructible
by Lemma~\ref{lem:all-computable}.
The problem now amounts to verifying if $\calG$ satisfies
the specification $\left(\bigwedge_{i \in \Agt} \Phi_{\mathsf{adm}}^i \right)
\implies \phi$. For all payoff functions, except mean-payoff, this can be reduced to
model checking an \LTL{} formula (since the measures are regular).
For $\overline{\mpfun}$ and $\underline{\mpfun}$, the result follows from \cite{BCHK-acm14}
which shows that the model checking problem against \ltlpayoff{} is
decidable.
\end{proof}

\smallskip
\noindent\textbf{Quantitative assume-admissible synthesis.}
In~\cite{brs15}, a new rule for reactive synthesis in non-zero sum $n$-player
games was proposed. The setting there is similar to the setting considered
here but it is Boolean: each player $i \in \Agt$ has his own omega-regular
objective $O_i \subseteq V^{\omega}$. The synthesis rule asks if player $i \in
\Agt$ has a strategy to enforce its own objective $O_i$ against
admissible strategies of the other players. In other words, the rule asks for
the existence of worst-case optimal strategies against rational adversaries.

The quantitative extension of this problem asks given a game $\calG$, a player
$i \in \Agt$, and a \ltlpayoff{} formula~$\phi$,
\(
	\exists \sigma \in \mathfrak{A}_i, \forall \tau \in \mathfrak{A}_{-i}, 
	\phi.
\)
Using Theorem~\ref{thm:ltl-characterization}, we can reduce this query to a
plain two-player zero-sum game on the game structure $\calG$ with objective:
\[ \textstyle
	\exists \sigma \in \Sigma_i,
	\forall \tau \in \Sigma_{-i}, 
    \Phi_{\mathsf{adm}}^i
    \land 
	\left(
    \left(\bigwedge_{j \in \Agt \setminus \{i\}} 
    \Phi_{\mathsf{adm}}^j
    \right)
    \implies
    \phi
    \right)
\]

Since for $\inffun$, $\supfun$, $\linffun$, $\lsupfun$, 
$\Phi_{\mathsf{adm}}^i$ and~$\phi$ are omega-regular,
the problem reduces to deciding the winner in a two-player zero-sum game with omega-regular
objectives. As a consequence, we obtain the following theorem:

\begin{theorem}\label{thm:aa-synthesis}
The quantitative assume-admissible synthesis problem for player~$i \in \Agt$ is
decidable for measures $\inffun$, $\supfun$, $\linffun$, $\lsupfun$.    
\end{theorem}

For the measures $\underline{\mpfun}$, $\overline{\mpfun}$, we obtain objectives
in which mean-payoff constraints and omega-regular constraints are mixed. On the
one hand, those objectives are outside known decidable classes of objectives
treated in~\cite{chj05} and in~\cite{crr13}. On the other hand, the
undecidability results obtained in~\cite{velner15} do not apply to them. This
motivates further research on zero-sum two player games with a mix of
mean-payoff and omega-regular objectives.

\smallskip
\noindent\textbf{Towards iterative elimination.}
Once we have computed the admissible strategies for each player, we restrict each player
to these strategies, and repeat the computation of the admissible strategies in
the restricted game.
This can be iterated several times and gives a process that
is called \emph{iterative elimination of dominated strategies}, and well known
in game theory.
This process is
difficult to analyze for mean-payoff, because objectives of
different players interfere in non-trivial ways and games with Boolean
combinations of mean-payoff objectives are undecidable~\cite{velner15}.  However
it seems feasible for regular payoffs, such as $\inffun$, $\supfun$, $\linffun$
and $\lsupfun$, for which we can construct parity automata recognizing outcomes
with $\pfunction{i} > q$.  
Given $i \ge 0$, we can actually compute a
parity automaton accepting the set of outcomes of $\mathcal{S}^i$ which is the set of
strategies that remain after~$i$ steps of elimination.
We summarize here the ingredients but leave the details for future work.
Assume we have a parity automaton representing the outcomes of
$\mathcal{S}^i$.  Note that for $i= 0$ this is simply all outcomes.
If the payoffs are regular, then we can compute values
$\cVal_i(h,\mathcal{S}^i)$, $\aVal_i(h,\mathcal{S}^i)$ and
$\acVal_i(h,\mathcal{S}^i)$, which correspond to cooperative, antagonistic, and
antagonist-cooperative values when players only play strategies from
$\mathcal{S}^i$.
We can then use these values as atomic propositions for a \ltlpayoff\ formulas
similar to $\Phi_{\mathsf{adm}}^i$ of Section~\ref{sec:ltl-characterization},
which characterizes outcomes of strategies of $\mathcal{S}^{i+1}$.  In the case
of regular payoffs this yields a parity automaton which represents the outcomes of
$\mathcal{S}^{i+1}$. This procedure can then be repeated to compute
outcomes that are possible under iterative elimination.

\bibliographystyle{plainurl}
\bibliography{biblio}

\clearpage

\section{Formal definition of the considered payoff functions}
We recall the definition of the measures below:
\begin{itemize}
	\item the $\inffun$ ($\supfun$) payoff, 
		is the minimum (maximum) weight seen along an outcome:
\(
	\inffun(\pi) = \inf\{ w(v_i,v_{i+1}) \st i \ge 1\}
\) and \(
	\supfun(\pi) = \sup\{ w(v_i,v_{i+1})\st i \ge 1\};
\)
	\item the $\linffun$ ($\lsupfun$) payoff,
		is the minimum (maximum) weight seen infinitely
		often:
\(
	\linffun(\pi) = \liminf_{i \to \infty} w(v_i,v_{i+1})
\) and \(
	\lsupfun(\pi) = \limsup_{i \to \infty} w(v_i,v_{i+1});
\)
	\item the \emph{mean-payoff} value, \ie~the limiting average
		weight, defined using $\liminf$ or $\limsup$ since
		the running averages might not converge:
\(
	\underline{\mpfun}(\pi) = \liminf_{k \to \infty} \frac{1}{k}
	\sum_{i = 1}^{k-1} w(v_i,v_{i+1}) 
\) and \(
	\overline{\mpfun}(\pi) = \limsup_{k \to \infty} \frac{1}{k}
	\sum_{i = 1}^{k-1} w(v_i,v_{i+1}).
\)
\end{itemize}

\section{Proof of Lem.~\ref{lemma:residual}}
\begin{proof}
  We have
  \[\begin{array}{ll}
      \aVal_i(\calG_h, h') &= \inf_{\tau} \sup_{\sigma} \pfunction{i}'(\outcome_{h'}(\calG_h, \sigma, \tau))\\
                           &= \inf_{\tau} \sup_{\sigma} \pfunction{i}(h \Cdot \outcome_{h'}(\calG_h, \sigma,\tau)).
    \end{array}
  \]
  Now for each fixed~$\tau$, 
  \[\begin{array}{l}
        \sup_{\sigma} \pfunction{i}(h \Cdot \outcome_{h'}(\calG_h, \sigma,\tau))
      = \sup_{\sigma} \pfunction{i}(\outcome_{h \Cdot h'}(\calG, \sigma,\tau)),
    \end{array}
  \]
  since for any~$\sigma$ on the left-hand side, one can define a strategy~$\sigma'$ by:
  $\sigma'(g) = \sigma(g_{\geq |h|})$ if~$h\Cdot h' \prefix g$, and defined arbitrarily otherwise.
  This proves that the LHS is less than or equal to the RHS.
  Conversely, for any~$\sigma$ on the right hand side, we can define~$\sigma'$ by
  $\sigma'(g) = \sigma(h\Cdot g)$ if~$h' \prefix g$ and arbitrarily otherwise.
  It follows that $\aVal_i(\calG_h, h') = \aVal_i(\calG,h\Cdot h')$.

  The cases for~$\acVal_i$ and~$\cVal_i$ are shown similarly.
\end{proof}

\section{Proof of Lem.~\ref{lem:sometimes-no-adm}}
\label{appendix:no-adm}

  \begin{proof}
	  As we noted above, the antagonistic (thus the cooperative) value at
	  vertex~$s_1$ is~$\infty$.  Given any strategy~$\sigma$, if~$s_3$ is
	  never visited, then~$\sigma$ is clearly weakly dominated by a strategy that
	  follows the outcome $(s_1s_2a)s_3^\omega$.  Otherwise, assume~$\sigma$
	  generates the outcome $h\cdot s_3^\omega$ where~$h$ does not
	  contain~$s_3$.  Then, $\sigma$ is weakly dominated by a strategy that
	  generates $h(s_2as_1)s_3^\omega$ whose payoff is greater than that
	  of~$\sigma$.  Thus, all strategies are weakly dominated.

    Let us now show that player~$1$ has no admissible strategy in~$\calG$.
    Consider any strategy~$\sigma$. Assume that $s_3$ is never visited
    by~$\sigma$ on compatible histories. In this case, the strategy that 
    goes to~$s_2$ once, and then goes to~$s_3$ dominates~$\sigma$. 
    In fact, its payoff is~$0$ in the worst-case, and for some adversary strategy,
    it yields a payoff of~$1$.
    Assume otherwise that at some history~$h$ with~$\last(h) = s_1$, $\sigma$ goes to~$s_3$.
    We have that
    \[
      \aVal_i(h) = |h|_a = \aVal_i(h,\sigma) = \cVal_i(h,\sigma) < \cVal_i(h) = \infty,
    \]
    where~$|h|_a$ is the number of~$a$'s in~$h$.
     Let us define~$\sigma'$ identically to~$\sigma$, except that at~$h$,
     it goes to~$s_2$, and at the next visit to~$s_1$, it goes to~$s_3$.
     Then $\sigma'$ dominates~$\sigma$. In fact, for all strategies~$\tau \in \Sigma_{-i}$,
     with $h \not \prefix \outcome(\sigma,\tau)$, we have $\outcome(\sigma',\tau) = \outcome(\sigma,\tau)$.
     For all $\tau \in \Sigma_{-i}$ for which~$h$ is a prefix of $\outcome(\sigma,\tau)$, the payoff is
     $\aVal_i(h)$. On the other hand, $\outcome(\sigma',\tau)$ has payoff at least~$\aVal_i(h)$,
     and for some particular~$\tau_0 \in \Sigma_{-i}$, the payoff of~$\outcome(\sigma',\tau_0)$ is
     $\aVal_i(h) +1$. 
     It follows that~$\sigma$ is weakly dominated.
  \end{proof}

\section{Proof of Lem.~\ref{lem:strongcoopop-exist}}
\label{appendix:strongcoopop-exist}
\begin{proof}
  We construct a strongly cooperative-optimal strategy
  for player~$i$, as follows.
  For each history~$h$ with $\aVal_i(h) < \cVal_i(h)$,
  let us fix an outcome~$\rho_h$ with $h \prefix \rho_h$ such that
  $\pfunction{i}(\rho_h) > \aVal_i(h)$. Such an outcome exists
  by Assumption~\ref{ass:optimal-coop-strats}; let us denote
  by~$(\sigma_h^\sco,\sigma_{h,-i}^\sco)$ any strategy profile compatible with
  $\rho_h$.
  Furthermore, for each~$h$
  with $\aVal_i(h) = \cVal_i(h)$, we fix a strategy~$\sigma_h^\wco$
  such that~$\aVal_i(h, \sigma_h^\wco) = \aVal_i(h)$, which exists
  by Assumption~\ref{ass:optimal-strats}. 
  Informally, we define our strategy~$\sigma$ as follows.
  If~$\aVal_i(h) = \cVal_i(h)$, then we switch to~$\sigma_h^\wco$.
  Otherwise, we start following~$\rho_h$, and whenever a player deviates
  from~$\rho_h$, say, at history~$h'$ with $h\prefix h'$, we start again
  according to whether~$\aVal_i(h') < \cVal_i(h')$.

     
    Let us now formalize the  strategy described in the core of the paper. In particular, we need to describe the set of histories
    at which the strategy switches to~$\rho_h$ or to~$\sigma^\wco_h$.
    We define \emph{decision points $\calD$} as a set of histories, where such a decision will be made,
    incrementally. We define~$\calD_i$, for each~$i$, such that~$\calD_0\cup\ldots\cup\calD_i$
    contains at least all decision points of length at most~$i$, and possibly some additional longer decision points.
    We will then let~$\calD = \cup_{i \geq 0} \calD_i$. Here, notice that for all~$j\geq i$, 
    $\{ h \in \history_\vinit, |h| \leq i \} \cap \calD_j$ is constant. So the union can be seen as a limit.

  Initially, for~$i=0$, $\vinit \in \calD_0$. Consider now $i>0$, and assume~$\calD_0,\ldots,\calD_i$ have been defined.
  For all~$h \in \calD_0\cup\ldots \cup\calD_i$ with~$|h|=i$, if~$\aVal_i(h) = \cVal_i(h)$, then~$\calD_{i+1}$ contains no history
  extending~$h$. If~$\aVal_i(h) < \cVal_i(h)$, then we add to~$\calD_{i+1}$ all histories of the following set:
  \[
    \begin{array}{c}
    \{ h' \mid h \prefix h' \land h'_{|h'|-1} \in V_{-i} \land h'_{\leq |h'|-1} \prefix \rho_h \land h' \not \prefix \rho_h \}\\
      \cup\\
      \{ h' \mid h \prefix h' \prefix \rho_h \land \aVal_i(h') = \cVal_i(h') )\}.
    \end{array}
  \]
  In other terms, all histories extending~$h$, and deviating from~$\rho_h$ by one step due to some player in~$-i$,
  and those prefixes of~$\rho_h$ where the antagonistic value equals the cooperative value.

  For any~$h$, let us define~$d(h)$ as the longest prefix of~$h$ that belongs to~$\calD$. This is well defined since~$\vinit \in \calD$.
  We now define our strategy~$\sigma$ as
  \[
    \sigma(h) = \left\{\begin{array}{ll}
                         \sigma_{d(h)}^\wco(h) & \text{if } \aVal_i(d(h)) = \cVal_i(d(h)),\\
                         \sigma_{d(h)}^\sco(h) & \text{if } \aVal_i(d(h)) < \cVal_i(d(h)).\\
                       \end{array}\right.
  \]
  We now show that
  \begin{claim}
	  $\sigma$ is SCO.
  \end{claim}
  The desired result follows.
  
  For any history~$h$ with~$\aVal_i(h) = \cVal_i(h)$, we have that~$\aVal_i(d(h)) =
  \cVal_i(d(h))$.  In fact, if~$d(h) = h$ then this trivially holds.  Otherwise,
  we have $d(h) \prefix h$. In this case, if~$\aVal_i(d(h)) < \cVal_i(d(h))$,
  then~$h \prefix \rho_{d(h)}$, and by definition of~$\calD$, $h \in \calD$,
  which is a contradiction.  Now, $\aVal_i(d(h), \sigma_{d(h)}^\wco) =
  \aVal_i(d(h)) = \cVal_i(d(h))$ by construction. Since~$h$ is compatible
  with~$\sigma_{d(h)}^\wco$, $\aVal_i(h, \sigma_{d(h)}^\wco) \geq
  \aVal_i(d(h))$, and since~$\cVal_i(\cdot)$ cannot increase
  along a history, we have $\cVal_i(h) \leq \cVal_i(d(h))$. We get that  
  \[
    \aVal_i(d(h)) \leq \aVal_i(h, \sigma_{d(h)}^\wco) \leq \aVal_i(h) \leq
    \cVal_i(h) \leq \cVal_i(d(h)).
  \]
  Since $\aVal_i(d(h)) = \cVal_i(d(h))$, we have that $\aVal_i(h,
  \sigma_{d(h)}^\wco) \leq \aVal_i(h)$.
  By definition of $\mathcal{D}$ we have that no extension of $d(h)$ is
  contained in $\mathcal{D}$ and that therefore $d(h') = d(h)$ for all
  histories $h'$ extending $h$. Thus $\sigma(h') = \sigma_{d(h)}^\wco(h')$ for
  all histories $h'$ extending $h$ and $\aVal_i(h,\sigma) = \aVal_i(h)$.

  Consider now history~$h$ with~$\aVal_i(h) < \cVal_i(h)$.
  If~$d(h) = h$, then 
  \[
	  \cVal_i(h,\sigma) = \cVal_i(h) > \aVal_i(h)
  \]
  by
  construction of~$\sigma$; and we have~\eqref{eqn:adm1}.  Assume otherwise.  We
  cannot have~$\aVal_i(d(h)) = \cVal_i(d(h))$, since the above sequence of
  inequalities again would prove that~$\aVal_i(h) = \cVal_i(h)$ which is not
  true. Thus, we must have~$\aVal_i(d(h)) < \cVal_i(d(h))$.  This means that~$h$
  is compatible with~$\sigma_{d(h)}^\sco$. Thus, we have
  \[
	  \pfunction{i}(\outcome_h(\sigma_{d(h)}^\sco, \sigma_{d(h),-i}^\sco)) =
  	\cVal_i(d(h)) > \aVal_i(d(h)).
  \]
  Now, since $h$ is a prefix of $\rho_{d(h)}$,
  we have
  \[
	  \cVal_i(h) \geq \pfunction{i}(\rho_{d(h)}) = \cVal_i(d(h)).
  \]
  It follows that~$\cVal_i(h,\sigma) = \cVal_i(d(h)) = \cVal_i(h)$
  since~$\cVal_i$ is non-increasing along histories.
\end{proof}

\section{Proof of Lem.~\ref{lem:negated-propadam}}
\label{appendix:negated-propadam}
\begin{proof}
  The negation of $\eqref{eqn:stlhd1}\lor\eqref{eqn:stlhd2}$ yields
  \[
  \begin{array}{l}
    \cVal_i(h,\sigma) > \aVal_i(h) \lor \aVal_i(h,\sigma) \geq \aVal_i(h)\\
    \text{and }\cVal_i(h,\sigma) = \aVal_i(h,\sigma) = \aVal_i(h) \implies
    \acVal_i(h) \leq \aVal_i(h).
   \end{array}
 \]
  We can rewrite the first line as follows
  \[
    \cVal_i(h,\sigma) > \aVal_i(h) \lor (\cVal_i(h,\sigma)\leq \aVal_i(h) \land
    \aVal_i(h,\sigma) \geq \aVal_i(h))\\
  \]
  The second term implies $\cVal_i(h,\sigma)\leq \aVal_i(h) \leq
  \aVal_i(h,\sigma)$, and is thus equivalent to~$\aVal_i(h,\sigma)=\cVal_i(h,\sigma)
  =\aVal_i(h)$. Using this, and distributing the conjunction over the
  disjunction, we get
  \[
    \begin{array}{l}
    \cVal_i(h,\sigma) > \aVal_i(h) \\
    \lor \big(
    \cVal_i(h,\sigma) = \aVal_i(h,\sigma) = \aVal_i(h) \land\\
    \qquad 
    \cVal_i(h,\sigma) = \aVal_i(h,\sigma) = \aVal_i(h) \implies
    \acVal_i(h) \leq \aVal_i(h).
    \big)
    \end{array}
  \]
  Simplifying the second term yields the equivalence with
  $\eqref{eqn:adm1}\lor\eqref{eqn:adm2}$.
\end{proof}

\section{Proof of Thm.~\ref{thm:ltl-characterization}}
\label{appendix:ltl-characterization}.
We first need the following lemma, which formalizes the following intuition: for
any outcome compatible with a given strategy, if $\acVal$ coincides with the
$\aVal$ at some position, and $\aVal$ does not decrease from that point on, then
$\acVal$ and $\aVal$ are constant in the rest of the outcome. A proof of the
claim is given in appendix.
\begin{lemma}
  \label{lem:acaVal-constant}
  Let $\calG$ be a game, $h \in \history_\vinit$ be a history, $\sigma_i$ a
  strategy of \player{i}, and $\rho$ an outcome extending $h$ compatible with
  $\sigma_i$, \ie~$\rho \in \outcome_h(\sigma_i)$.  Assume there exists $j \ge
  |h|$ such that $\acVal_i(\rho_{\le j}) = \aVal_i(\rho_{\le j})$.  The
  following hold
  \begin{enumerate}[label=(a)]
	  \item for all $k > j$, if for all $j \le j' < k$ we have that
		  $\rho_{j'} \in V_i \implies \aVal_i(\rho_{\le j'+1}) \ge
		  \aVal_i(\rho_{\le j'})$ then $\acVal(\rho_{\le k}) =
		  \aVal(\rho_{\le k}) =
		  \acVal(\rho_{\le j}) = \aVal(\rho_{\le j})$;
	  \item if $\aVal_i(\rho_{\leq j},\sigma_i) =
		  \aVal_i(\rho_{\le j})$ then
		  $\pfunction{i}(\rho) = \aVal_i(\rho_{\le j})$
		  and for all $k \ge j$
		  we have $\acVal(\rho_{\le k}) =
		  \aVal(\rho_{\le k}) = \acVal(\rho_{\le j}) =
		  \aVal(\rho_{\le j})$.
  \end{enumerate}
\end{lemma}
\begin{proof}
\item \subparagraph{(a)}
	Consider an arbitrary $k > j$. We first note that the players $-i$
	cannot decrease $\aVal$; more precisely, it follows form the definition
	of $\aVal$ that for any $j'$ with $\rho_{j'} \not \in V_i$,
	$\aVal_i(\rho_{\leq j'+1}) \ge \aVal_i(\rho_{\leq j'})$.  Thus, if for
	all $j \le j' < k$, $\rho_{j'} \in V_i \implies \aVal_i(\rho_{\leq j'+1}) \ge
	\aVal_i(\rho_{\leq j'})$, then we can write
	\begin{equation}\label{eqn:hypo1}
		\forall j \le j < k,  \aVal_i(\rho_{\leq j'+1}) \ge
		\aVal_i(\rho_{ \leq j'}).
	\end{equation}

	We will now argue that $\acVal_i(\rho_{\leq j'}) \le \acVal_i(\rho_{\leq
	j})$ for all $j \le j' < k$. Suppose, towards a contradiction, that
	there is some $j \le j' < k$ such that $\acVal_i(\rho_{\leq j'}) >
	\acVal_i(\rho_{\leq j})$.  We are going to construct a
	strategy~$\sigma''_i$ for \player{i} which witnesses that
	$\acVal_i(\rho_{\leq j})\geq \acVal_i(\rho_{\leq j'})$ showing a
	contradiction.

	We are going to define a strategy profile $(\sigma_i', \sigma_{-i}')$ by
	distinguishing two cases.
    	\begin{enumerate}
		\item If~$\acVal_i(\rho_{\leq j'}) > \aVal_i(\rho_{\leq j'})$,
			then let $(\sigma_i',\sigma'_{-i})$ be a strategy
			profile, and~$\epsilon>0$ such that
        \[
		\pfunction{i}(\outcome_{\rho_{\leq j'}}(\sigma_i',\sigma_{-i}'))
		\geq \acVal_i(\rho_{\leq j'})-\epsilon > \max(\aVal_i(\rho_{\leq
		j'}),\acVal_i(\rho_{\leq j})),
        \]
			while satisfying $\aVal_i(\rho_{\leq j'}, \sigma_i') =
			\aVal_i(\rho_{\leq j'})$.  Such a strategy profile
			exists by the definition of~$\acVal_i$.
		\item If~$\acVal_i(\rho_{\leq j'}) = \aVal_i(\rho_{\leq j'})$,
			then define $\sigma_i'$ as a worst-case optimal strategy
			from history~$\rho_{\leq j'}$ (by
			Assumption~\ref{ass:optimal-strats}), and
			choose~$\sigma_{-i}'$ arbitrarily.        
	\end{enumerate}

	Let~$\sigma^\wco_{i,h}$ be a \player{i} strategy satisfying $\aVal_i(h,
	\sigma^\wco_{i,h}) = \aVal_i(h)$, which exists by
	Assumption~\ref{ass:optimal-strats}.  Let $\sigma''_i$ denote the
	strategy of \player{i} which, from~$\rho_{\leq j}$, follows the
	history~$\rho_{j}\rho_{j+1}\ldots \rho_{j'}$ and switches
	to~$\sigma_i'$, and at any other prefix~$h$ switches
	to~$\sigma^\wco_{i,h}$.  Formally, let us define 
    	\[
    		\sigma_i''(h) = \left\{
        	\begin{array}{ll}
		\rho_{l} & \text{if } h = \rho_{\leq l-1}, l \leq j', \rho_{l-1}
		\in V_i,\\ \sigma_i'(h) & \text{if } \rho_{\leq j'} \prefix h,
		\last(h) \in V_i\\ \sigma_{i,h_1}^\wco(h) & \text{if } h =
		h_1h_2, |h_2|\geq 1, h_1 = \lcp(h,\outcome_{\rho_{\leq
		j'}}(\sigma'_i,\sigma'_{-i})).
      		\end{array}
    		\right.
	\]
	We define similarly,
    	\[
    		\sigma_{-i}''(h) = \left\{
		\begin{array}{ll}
		\rho_{l} & \text{ if } h = \rho_{l-1}, l \leq j',\rho_{l-1} \in
		V_{-i},\\ \sigma_{-i}'(h) & \text{ if } \rho_{\leq j'} \prefix
		h, \last(h) \in V_{-i},\\ \sigma^0_{-i}(h) & \text{ otherwise},
		\end{array}
		\right.
 	\]
	where~$\sigma^0_{-i}$ is an arbitrary strategy profile for~$-i$.

	Now, by construction, $\outcome_{\rho_{\leq
	j}}(\sigma''_i,\sigma''_{-i}) = \outcome_{\rho_{\leq j}}(\sigma_i',
	\sigma_{-i}')$ which has payoff at least $\acVal_i(\rho_{\leq j'}) -
	\epsilon > \acVal_i(\rho_{\leq j})$ in the first case, and equal
	to~$\acVal_i(\rho_{\leq j'}) > \acVal_i(\rho_{\leq j})$ in the second
	case.  Moreover, against any other strategy~$\tau$ of players~$-i$, we
	have 
	\[
		\pfunction{i}(\outcome_{\rho_{\leq j}}(\sigma''_i,\tau)) \geq
		\aVal_i(\rho_{\leq j})
	\]
	since \player{i} switches to $\sigma^\wco_{i,h}$
	at any minimal history~$h$ that is not a prefix of the outcome
	of~$(\sigma_i',\sigma_{-i}')$; and for any such~$h$, $\aVal_i(h) \geq
	\aVal_i(\rho_{\leq j})$. In fact, we saw above that $\aVal_i(h_{\leq
	|h|-1}) \geq \aVal_i(\rho_{\leq j})$. Moreover, since $h_{\leq |h|-1}$
	ends in a vertex in~$V_{-i}$, the antagonistic value cannot decrease:
	so~$\aVal_i(h) \geq \aVal_i(\rho_{\leq j})$.

	This is a contradiction since from history~$\rho_{\leq j}$, $\sigma_i''$
	achieves a cooperative value greater than~$\acVal_i(\rho_{\leq j})$, and
	a worst-case value of at least~$\aVal_i(\rho_{\leq j})$.    

	To conclude, we have shown that for all $j \le j' < k$ both
	$\aVal_i(\rho_{\leq j}) \le \aVal_i(\rho_{\leq j'})$ and
	$\acVal_i(\rho_{\leq j}) \ge \acVal_i(\rho_{\leq j'})$. Since, by
	definition $\aVal(\rho_{\leq j'}) \le \acVal_i(\rho_{\leq j'})$, for all
	$j' \ge j$, and $k$ was chosen arbitrarily, the result follows.

\subparagraph{(b)}
	From~$\aVal_i(\rho_{\leq j},\sigma_i) = \aVal_i(\rho_{\leq j})$,
	and~$\acVal_i(\rho_{\leq j}) = \aVal_i(\rho_{\leq j})$, it follows that
	all outcomes of~$\sigma_i$ that extend~$\rho_{\leq j}$ must have payoff
	exactly~$\aVal_i(\rho_{\leq j})$. In particular, $\pfunction{i}(\rho) =
	\aVal_i(\rho_{\leq j})$.
    
	Let us show that~$\aVal_i(\rho_{\leq k}) = \acVal_i(\rho_{\leq k}) =
	\aVal_i(\rho_{\leq j})$ for all~$k> j$.  Note that since $j \leq k$,
	$\aVal(\rho_{\leq k},\sigma_i) \ge \aVal(\rho_{\leq j},\sigma_i)$.
	Since $\aVal_i(\rho_{\leq k}, \sigma_i) \le \aVal_i(\rho_{\leq
	k})$, we have
    	\[
		\aVal_i(\rho_{\leq k}) \geq \aVal_i(\rho_{\leq k}, \sigma) \geq
		\aVal_i(\rho_{\leq j}, \sigma) = \aVal_i(\rho_{\leq j}).
    	\]
	Towards a contradiction, assume that this is not an equality.  We cannot
	have $\aVal_i(\rho_{\leq k}) < \aVal_i(\rho_{\leq j})$ since, this would
	contradict $\aVal_i(\rho_{\leq j}, \sigma) = \aVal_i(\rho_{\leq j})$.
	Also, if we assume that $\aVal_i(\rho_{\leq k}) > \aVal_i(\rho_{\leq j})$,
	since~$\rho_{\leq k}$ is compatible with~$\sigma_i$ from
	history~$\rho_{\leq j}$, it follows that~$\acVal_i(\rho_{\leq j}) \geq
	\aVal_i(\rho_{\leq k})$, contradiction.

    	The result follows.
\end{proof}

We can now prove Thm.~\ref{thm:ltl-characterization}.
\begin{proof}[Proof of Thm.~\ref{thm:ltl-characterization}]
	\fbox{$\Leftarrow$} 
	Suppose that $\rho$ is compatible with an admissible strategy $\sigma_i$
	for \player{i}. We will show that for any prefix $h$ of $\rho$, such
	that $\last(h) \in V_i$, $\rho,|h|$ satisfies 
	either $\phi_{\ref{eqn:adm1}}$ or
	$\phi_{\ref{eqn:adm2}}$.
	
	First, observe that $(\rho,|h|) \models \aValProp_q$ for exactly one $q \in
	\mathsf{aValues}_i$. Now, since $\sigma_i$ is admissible, we have, from
	Thm.~\ref{thm:lhd-admissibles}, that at $h$ either~\eqref{eqn:adm1}
	or~\eqref{eqn:adm2} holds.
	\begin{itemize}
		\item If~\eqref{eqn:adm1} holds, then $\cVal_i(h,\sigma) >
			\aVal_i(h)$, so there is $\sigma_{-i} \in \Sigma_{-i}$
			compatible with $h$ such that $\pfunction{i}(\rho') >
			\aVal_i(h) = q$ where $\rho' =
			\outcome_h(\calG,\sigma_i,\sigma_{-i})$.  If $\rho
			= \rho'$ then $(\rho,|h|) \models
			\phi_{\ref{eqn:adm1}}$.  Otherwise, let $j$ be the first
			index where $\rho_j \ne \rho'_j$. We have $\rho_{j-1} \in
			V_{-i}$ and $j > |h|$ because $\sigma_i$ and
			$\sigma_{-i}$ are compatible with $h$.  Moreover
			$\cVal_i(\rho_{j-1}) \ge \pfunction{i}(\rho') > q$, so
			$\gAlt_q$ is satisfied by $\rho,j$. Therefore the
			sub-formula $\F(\gAlt_q)$, and thus
			$\phi_{\ref{eqn:adm1}}$, are satisfied by $\rho,|h|$.
		\item Otherwise~\eqref{eqn:adm2} holds, which means that
			$\aVal_i(h,\sigma) = \cVal_i(h,\sigma) = \acVal_i(h) =
			\aVal_i(h) = q$.  By Lem.~\ref{lem:acaVal-constant}
			(b), we have that $\pfunction{i}(\rho) = q$
			and~$\aVal_i$ is constantly equal to~$q$ along $\rho$
			after $h$; so~$\phi_{\ref{eqn:adm2}}$ is satisfied by
			$\rho,|h|$.
	\end{itemize}

	\fbox{$\Rightarrow$}
	Let $\rho$ be an outcome which satisfies $\Phi_{\textsf{adm}}^i$.  We
	define strategy~$\sigma_i$ for \player{i} which follows~$\rho$, and on
	any history that is not a prefix of~$\rho$, switches immediately either
	to a strongly cooperative-optimal or to a worst-case optimal strategy,
	depending on whether $\phi_1$ or $\phi_2$ holds.  Formally, 
    for each history~$h$, 
	let us fix a strongly cooperative-optimal strategy~$\sigma_h^\sco$ 
    for \player{i} in the game~$\calG_h$, 
	a worst-case optimal strategy~$\sigma^\wco_h$
	from~$h$ in~$\calG$, that is, $\aVal_i(h, \sigma^\wco_h) = \aVal_i(h)$.
	Remark the former is guaranteed to exist because of
	Lem.~\ref{lem:residual-admissible}.
	Given $h$ such that $\last(h) \in V_i$, we let
	$\sigma_i(h)$ be
	\[
		\begin{cases}
			\rho_{k+1} & \text{if } h = \rho_{\le k}\\

			\sigma^\sco_{h_0}(h') & \text{if } h_0 \models\phi_1,
              \text{$h_0$ is the smallest prefix of~$h$ with } \last(h) \in V_i \land  h_0 \not \prefix \rho,  \\
              & \text{ and } h' = h_{\geq |h_0|}.\\

			\sigma^\wco_{h_0}(h) &\text{if } h_0 \models
			\lnot\phi_1 \land \phi_2, \text{$h_0$ is the smallest prefix of~$h$ with } \last(h) \in V_i \land  h_0 \not \prefix \rho.\\
		\end{cases}
	\]
	It is clear that~$\rho$ is compatible with~$\sigma_i$.  We now show
	that~$\sigma_i$ is admissible.

	Let us start with histories $h$ that are not a prefix of~$\rho$,
    say, of the form~$h=h_0 \Cdot h'$ where~$h_0$ is the smallest prefix of~$h$ ending in~$V_i$ which is not a prefix of~$\rho$.
    We show that either~\eqref{eqn:adm1} or~\eqref{eqn:adm2} hold for~$h$ and~$\sigma_i$
	by distinguishing two cases.
	\begin{enumerate}
	\item Assume $\sigma_i$ follows $\sigma^\sco_{h_0}$ at~$h_0$, which is
		strongly cooperative-optimal, thus admissible in the game~$\calG_{h_0}$ by Lem.~\ref{lem:sco-are-adm}.
        This strategy thus satisfies~\eqref{eqn:adm1} at~$h'$, that is,
        \begin{equation}
          \label{eqn:resid-adm1}
          \aVal_i(\calG_{h_0}, h') < \cVal_i(\calG_{h_0}, h', \sigma^\sco_{h_0}),
        \end{equation}
        or it satisfies~\eqref{eqn:adm2}, which is 
        \begin{equation}
          \label{eqn:resid-adm2}
          \aVal_i(\calG_{h_0}, h') =\acVal_i(\calG_{h_0},h') = \aVal_i(\calG_{h_0}, h', \sigma^\sco_{h_0}).
        \end{equation}
        We are going to show that~\eqref{eqn:adm1} or~\eqref{eqn:adm2} hold for~$\calG$ at history~$h$.
        
        Assume~\eqref{eqn:resid-adm1} holds at~$h'$ in~$\calG_{h_0}$. By Lem.~\ref{lemma:residual}, we have
        that the LHS of~\eqref{eqn:resid-adm1} is equal to $\aVal_i(\calG, h_0 \Cdot h') = \aVal_i(\calG, h)$.
        But, by the definition of~$\sigma_i$, the RHS of~\eqref{eqn:resid-adm1} is equal to
        $\cVal_i(\calG, h_0 \Cdot h', \sigma_i)$ since~$\sigma_i(h_0\Cdot g) = \sigma^\wco_{h_0}(g)$ for all
        $g \in \history_{\last(h_0)}(\calG)$.
        
        Assume~\eqref{eqn:resid-adm2} holds at~$h'$ in~$\calG_{h_0}$. 
        By Lem.~\ref{lemma:residual}, we have that $\aVal_i(\calG_{h_0},h') = \aVal_i(\calG, h)$
        and~$\acVal_i(\calG_{h_0},h') = \acVal_i(\calG, h)$. Moreover,
        $\aVal_i(\calG_{h_0}, h', \sigma^\sco_{h_0}) = \aVal_i(\calG, h, \sigma_i)$
        since~$\sigma_i(h_0 \Cdot g) = \sigma^\sco_{h_0}(g)$ for all~$g \in \history_{\last(h_0)}(\calG)$.

        Thus~\eqref{eqn:adm1} or~\eqref{eqn:adm2} hold in~$\calG$ for~$\sigma$ and~$h$.

	\item Or $\sigma$ follows $\sigma^\wco_{h_0}$ which means,
      for $k = |h_0|$,
		$(\rho,k) \models \phi_{\ref{eqn:adm2}}$, and so for some
		$q \in \mathsf{aValues}_i$ we have that $(\rho,k)
		\models \aValProp_q \land
		\acValProp_q$.  Therefore $\aVal_i(\rho_{\le k}) =
		\acVal_i(\rho_{\le k})$.  By construction of $\sigma_i$ we have
		that
		$\aVal_i(\rho_{\leq k}, \sigma_i) = \aVal_i(\rho_{\leq k},
		\sigma^\wco_{\rho_{\leq k}}) = \aVal_i(\rho_{\leq k})$. It is
		clear from the definition of $\cVal_i$ that $\aVal_i(\rho_{\leq
		k},\sigma_i) \le \cVal_i(\rho_{\leq k}, \sigma_i)$. We claim
		that, in fact, we have equality in this case. Towards, a
		contradiction, assume that this is not the case. Thus, it holds
		that $\aVal_i(\rho_{\leq k},\sigma_i) < \cVal_i(\rho_{\leq k},
		\sigma_i)$. Since we have already established that
		$\aVal_i(\rho_{\leq k},\sigma_i) = \aVal_i(\rho_{\leq k})$, this
		implies (by def. of $\acVal_i$) that $\acVal_i(\rho_{\leq k}) >
		\aVal_i(\rho_{\leq k})$.  Contradiction.  Hence \eqref{eqn:adm2}
		holds.
	\end{enumerate}

	For any history~$\rho_{\leq k}$ with~$\rho_k \in V_i$, we consider two
	cases.
	\begin{itemize}
		\item If $(\rho,k) \models \phi_{\ref{eqn:adm1}}$  then
			either $\pfunction{i}(\rho) > \aVal_i(\rho_{\leq k})$,
			and~\eqref{eqn:adm1} holds, or $(\rho,k') \models 
			\gAlt_q$ for some
			position~$k'>k$ and for $q = \aVal_i(\rho_{\geq
			k})$.  In the latter
			case,~$\rho_{k'} \in V_{-i}$ and $\cVal_i(\rho_{k'}) >
			\aVal_i(\rho_k)$, so for any history~$\rho_{\leq k}
			\cdot v$ with $v\ne \rho_{k'+1}$, by construction,
			$\sigma_i$ switches to a strongly cooperative-optimal
			strategy; thus $\cVal_i(\rho_{\leq k'} \cdot v,
			\sigma_i) > \aVal_i(h)$. It then follows
			that~$\cVal_i(h, \sigma_i) \geq \cVal_i(\rho_{\leq
			k'}\cdot v,\sigma_i) > \aVal_i(h)$, which
			satisfies~\eqref{eqn:adm1}.
	  \item Otherwise $(\rho,k) \models \phi_{\ref{eqn:adm2}} \land \lnot
		  \phi_{\ref{eqn:adm1}}$.  In
		  particular, since $\phi_2$ is satisfied, there is some $q$ such
		  that~$\pfunction{i}(\rho) = q = \aVal_i(\rho_{\leq k}) =
		  \acVal_i(\rho_{\leq k})$ and $\rho_{\geq k} \models
		  \G(\aValProp_q)$.  Hence the antagonistic values of all
		  histories $h'$ extending $\rho_{\ge k}$
		  are all equal to~$\aVal_i(\rho_{\leq
		  k})$.  Since $\sigma_i$ immediately switches to a worst-case
		  optimal strategy when we do not follow $\rho$, all outcomes
		  of~$\sigma_i$ extending~$\rho_{\leq k}$ have payoff at
		  least~$\aVal_i(\rho_{\leq k})$, so $\aVal_i(\sigma_i,
		  \rho_{\leq k}) = \aVal_i(\rho_{\leq k})$.  Moreover outcomes
		  of $\sigma_i$ from $\rho_{\leq k}$ has coinciding cooperative
		  and antagonistic values. That is, we necessarily have that
		  $\cVal_i(\rho_{\le k},\sigma_i) = \aVal_i(\rho_{\le
		  k},\sigma_i)$. Indeed, if this were not the case, then the
		  cooperative value would be strictly higher and this would
		  contradict the fact that $\aVal_i(\rho_{\leq k}) =
		  \acVal_i(\rho_{\leq k})$---this is by definition of
		  $\acVal_i$.  This shows $\aVal_i(\rho_{\le k},\sigma_i) =
		  \cVal_i(\rho_{\le k},\sigma_i) = \aVal_i(\rho_{\le k}) =
		  \acVal_i(\rho_{\le k})$, hence~\eqref{eqn:adm2} holds.\qed
	\end{itemize}
	\let\qed\relax
\end{proof}




%
\section{Making infimum and supremum functions prefix-independent}
\label{app:prefix-indep}
The payoff functions defined using $\inffun$ and $\supfun$ are not
prefix-independent.  Nevertheless, our results are also applicable to those
measures after applying a simple and classical transformation to the game
structure. This transformation ensures prefix-independence for those measures
for all plays in the new game and a bijection between the strategies in the
original game and the strategies in the modified game. For $\inffun$ (rest.
$\supfun$), the transformation is as follows: for each player $i \in \Agt$, we
record, as additional information in the vertices of the game, the minimal
(resp. maximal) value seen so far with the weight function $w_i$. We then modify
the weight function to output the recorded value if the original label of the
edge that is taken is larger (resp. smaller) than or equal to the recorded
value, and to output the original value otherwise. In the latter case, the new
minimum (resp.  maximum) is recorded. Clearly, the measure $\inffun$ (resp.
$\supfun$) is prefix-independent on all the plays of the new game structure.

\section{Computing the antagonistic-cooperative value}
\label{app:acval}
Under Assumptions~\ref{ass:optimal-strats}-\ref{ass:prefix-indep},
and assuming values~$\aVal_i$ and~$\cVal_i$ can be computed at each vertex,
we show how~$\acVal_i$ can be computed as well.
Given game~$\calG$, history~$h$, and \player{i}, let $\Gaih$ denote
the game obtained by restricting $\calG$ to vertices~$v$ such that $\aVal_i(v)
\ge \aVal_i(h)$. One can check that  all vertices in $\Gaih$ have at least one
outgoing edge and that $\acVal_i(h) = \cVal_i(\Gaih, \last(h))$. 
This is formalized and proven below.
\begin{lemma}\label{lem:acval-subgraph-def}
	 For all games defined using payoff functions considered here,
	\begin{inparaenum}[a)]
		\item any vertex in $\Gaih$ reachable from $\last(h)$ has at
			least one outgoing edge;
		\item for any \player{i} and history~$h$,
			$\cVal_i(\Gaih, \last(h)) = \acVal_i(\calG, h)$.
	\end{inparaenum}
\end{lemma}
\begin{proof}
\item \subparagraph{Proof of first item.}
	Let $v$ be a vertex in $\Gaih$ reachable from $\last(h)$. Because of
	Assumption~\ref{ass:optimal-strats}, we have that in $\calG$ there is a
	strategy $\sigma_i \in \Sigma_i$ which achieves a payoff of at least
	$\aVal_i(v)$. By construction of $\Gaih$, we also have that $\aVal_i(v)
	\ge \aVal_i(h)$.  Now, if $v \in V_i$, the edge $(v,\sigma_i(v))$ is
	also present in $\Gaih$. Otherwise, $\sigma_i(v)$ from $\calG$ is such
	that $\aVal_i(\sigma_i(v)) < \aVal_i(h)$ and this contradicts our choice
	of $\sigma_i$.  If $v \not\in V_i$, then, by the same argument, all
	outgoing edges from $v$ in $\calG$ should be present in $\Gaih$.

\item \subparagraph{Proof of second item.}
	We observe that for all strategies $\sigma_i \in \Sigma_i$, if $\rho$ is
	an outcome compatible with $\sigma_i$ from $\last(h)$ and $\rho$ visits
	a vertex $v$ not in $\Gaih$, then
	\(
		\aVal_i(\calG,h,\sigma_i) < \aVal_i(\calG,h).
	\)
	Indeed, since $v$ is not in $\Gaih$, we have that $\aVal_i(\calG,v) <
	\aVal_i(\calG,h)$. The claim then follows from
	Assumption~\ref{ass:prefix-indep} and the definition of $\aVal_i$.
	Thus the outcomes affecting
	the value $\acVal_i$ are only those which stay in $\Gaih$. That is, we
	have that
	\[
		\acVal_i(\calG,h)=\acVal_i(\Gaih,\last(h)).
	\]
	By definition of $\cVal_i$ we then get that 
	\begin{equation}
		\cVal_i(\Gaih,\last(h)) \ge \acVal_i(\Gaih,\last(h)) =
		\acVal_i(\calG,h).
	\end{equation}

	Consider now an outcome $\rho$ in $\Gaih$ from $\last(h)$ witnessing the
	cooperative value, that is $\pfunction{i}(\rho) =
	\cVal_i(\Gaih,\last(h))$. Recall the existence of $\rho$ is implied by
	Property~\ref{ass:optimal-coop-strats}. Consider a strategy $\sigma_i$
	for \player{i} which follows $\rho$ until the players $-i$ stops
	following it (that is, until the current history is no longer a prefix
	of $\rho$). At this point, $\sigma_i$ switches to a strategy which
	ensures at least $\aVal_i(\calG,h)$. It is easy to see that $\sigma_i$
	is such that $\cVal_i(\Gaih,\last(h),\sigma_i) =
	\cVal_i(\Gaih,\last(h))$ and, further, $\aVal_i(\Gaih,\last(h),\sigma_i)
	= \aVal_i(\Gaih,\last(h))$. The latter implies that
	\begin{equation}
		\cVal_i(\Gaih,\last(h)) \le \acVal_i(\Gaih,\last(h)) =
		\acVal_i(\calG,h)
	\end{equation}
	which concludes the proof.
\end{proof}


\section{Proof of Thm.~\ref{thm:mc-admissible}}
\begin{proof}
For each player~$i \in \Agt$, consider the formula $\Phi_{\mathsf{adm}}^i$
from Thm.~\ref{thm:ltl-characterization}, which describes
the set $\outcome(\calG, \admstratset_i(\calG))$.
To compute $\Phi_{\mathsf{adm}}^i$, 
we just need to compute $\aVal_i(v)$, $\cVal_i(v)$ and $\acVal_i(v)$ for
each vertex $v$ and player $i \in \Agt$.  This is possible due to
Lem.~\ref{lem:all-computable}.

It is easy to see that the problem now amounts to verifying if $\calG$ satisfies
the specification $\left(\bigwedge_{i \in \Agt} \Phi_{\mathsf{adm}}^i \right)
\implies \phi$.

For payoff functions $\inffun$, $\supfun$, $\lsupfun$, and $\linffun$, all
propositions of the form $\pfunction{i} \bowtie v$ are regular; in particular,
each of them can be replaced by an \LTL{} formula. The problem of model checking
under admissibility is then reduced to model checking the obtained \LTL{}
formula.

For payoff functions $\overline{\mpfun}$ and $\underline{\mpfun}$, we use
the results of~\cite{BCHK-acm14} which show that the model checking problem
against a formula in \ltlpayoff{} is decidable.
\end{proof}

\end{document}